\documentclass[onefignum,onetabnum]{siamart171218}

\newcommand{\C}{\mathcal{C}}

\usepackage{lipsum}
\usepackage{amsfonts}
\usepackage{graphicx}
\usepackage{epstopdf}
\usepackage{algorithmic}
\ifpdf
  \DeclareGraphicsExtensions{.eps,.pdf,.png,.jpg}
\else
  \DeclareGraphicsExtensions{.eps}
\fi

\newsiamremark{remark}{Remark}
\newsiamremark{hypothesis}{Hypothesis}
\crefname{hypothesis}{Hypothesis}{Hypotheses}
\newsiamthm{claim}{Claim}

\headers{On Independent Cliques and  Linear Complementarity Problems}{K.N. Chadha and A.A. Kulkarni}

\title{On Independent Cliques and  Linear Complementarity Problems\thanks{Submitted to the editors DATE.}}

\author{Karan N. Chadha\thanks{Department of Electrical Engineering, Indian Institute of Technology Bombay,
  (\email{karanchadhaiitb@gmail.com}).}
\and Ankur A. Kulkarni\thanks{Systems and Control Engineering, Indian Institute of Technology Bombay, 
  (\email{kulkarni.ankur@iitb.ac.in}).}}

\usepackage{amsopn}

\makeatletter
\newcommand*{\addFileDependency}[1]{% argument=file name and extension
  \typeout{(#1)}% latexmk will find this if $recorder=0 (however, in that case, it will ignore #1 if it is a .aux or .pdf file etc and it exists! if it doesn't exist, it will appear in the list of dependents regardless)
  \@addtofilelist{#1}% if you want it to appear in \listfiles, not really necessary and latexmk doesn't use this
  \IfFileExists{#1}{}{\typeout{No file #1.}}% latexmk will find this message if #1 doesn't exist (yet)
}
\makeatother

\newcommand*{\myexternaldocument}[1]{%
    \externaldocument{#1}%
    \addFileDependency{#1.tex}%
    \addFileDependency{#1.aux}%
}
\usepackage{amsmath, amssymb, xspace,url}
\usepackage{epsfig}
\usepackage{longtable}
\usepackage{color}
\def\bfone{{\bf 1}}
\def\bfe{{\bf e}}
\def\t{^\top}
\def\th{^{\rm th}}

\def\bkE{{\rm I\kern-.17em E}}
\def\bk1{{\rm 1\kern-.17em l}}
\def\bkD{{\rm I\kern-.17em D}}
\def\bkR{{\rm I\kern-.17em R}}
\def\bkP{{\rm I\kern-.17em P}}

\def\bkZ{{\bf{Z}}}
\def\bkE{{\rm I\kern-.17em E}}
\def\bk1{{\rm 1\kern-.17em l}}
\def\bkD{{\rm I\kern-.17em D}}
\def\bkR{{\rm I\kern-.17em R}}
\def\bkP{{\rm I\kern-.17em P}}
		\def\bkE{{\rm I\kern-.17em E}}
		\def\bk1{{\rm 1\kern-.17em l}}
		\def\bkD{{\rm I\kern-.17em D}}
		\def\bkR{{\rm I\kern-.17em R}}
		\def\bkP{{\rm I\kern-.17em P}}
		\def\bkY{{\bf \kern-.17em Y}}
		\def\bkZ{{\bf \kern-.17em Z}}
		\def\bkC{{\bf  \kern-.17em C}}

\def\Kscr{{\cal K}}

\def\SOL{{\rm SOL}}

\def\LCP{{\rm LCP}}

\let\forallnew\forall
\renewcommand{\forall}{\forallnew\ }
\let\forall\forallnew

\def\b12{(\beta_1,\beta_2)}

\def\subject{\hbox{\rm subject to}}
\newcommand{\Real}{\ensuremath{\mathbb{R}}}

\def\Cbar{\bar C}

\def\Gbar{\bar G}

\def\hbar{\skew{4.2}\bar h}

\def\Nbar{\skew{4.4}\bar N}

\def\Shat{\widehat S}
\def\Sbar{\skew2\bar S}

\def\Vbar{\skew2\bar V}

\def\superstar{^{\raise 0.5pt\hbox{$\nthinsp *$}}}
\def\SUPERSTAR{^{\raise 0.5pt\hbox{$*$}}}
\def\lamstarT {\lambda^{\raise 0.5pt\hbox{$\nthinsp *$}T}}

\newlength{\noteWidth}
\long\def\notes#1{\ifinner
{\tiny #1}
\else
\marginpar{\parbox[t]{\noteWidth}{\raggedright\tiny #1}}
\fi\typeout{#1}}
 \def\notes#1{\typeout{read notes: #1}} %uncomment for final version

\newcommand{\ie}{i.e.\@\xspace} %%% i.e.,
\newcommand{\maximize}[1]{\displaystyle\maxim_{#1}}
\newcommand{\maxim}{\mathop{\hbox{\rm maximize}}}

\def\norm#1{\|#1\|}
\def\spose#1{\hbox to 0pt{#1\hss}}

\def\text #1{\hbox{\quad#1\quad}}

\makeatletter
\newcommand{\pushright}[1]{\ifmeasuring@#1\else\omit\hfill$\displaystyle#1$\fi\ignorespaces}
\newcommand{\pushleft}[1]{\ifmeasuring@#1\else\omit$\displaystyle#1$\hfill\fi\ignorespaces}
\makeatother

\def\nthinsp{\mskip -2   mu}

\newcounter{example}
\renewcommand{\theexample}{\thesection.\arabic{example}}

{\begin{list}{}%
         {\setlength{\leftmargin}{#1}}%
         \item[]%
}
{\end{list}}

		\def\bsp{\begin{split}}
		\def\beq{\begin{eqnarray}}
		\def\bal{\begin{align*}}
		\def\bc{\begin{center}}
		\def\be{\begin{enumerate}}
		\def\bi{\begin{itemize}}
		\def\bs{\begin{small}}
		\def\bS{\begin{slide}}
		\def\ec{\end{center}}
		\def\ee{\end{enumerate}}
		\def\ei{\end{itemize}}
		\def\es{\end{small}}
		\def\eS{\end{slide}}
		\def\eeq{\end{eqnarray}}
		\def\eal{\end{align*}}
		\def\esp{\end{split}}
		\def\qed{ \vrule height7.5pt width7.5pt depth0pt}  %width4.17pt depth0pt} 

\def\maxproblem#1#2#3#4{\fbox
		 {\begin{tabular*}{0.80\textwidth}
			{@{}l@{\extracolsep{\fill}}l@{\extracolsep{6pt}}l@{\extracolsep{\fill}}c@{}}
				#1 & $\maximize{#2}$ & $#3$ & $ $ \\[5pt]
					 & $\subject\ $    & $#4$ & $ $
			\end{tabular*}}
			}

\def\maxfourproblem#1#2#3#4#5#6#7{\fbox
		 {\begin{tabular*}{0.80\textwidth}
			{@{}l@{\extracolsep{\fill}}l@{\extracolsep{6pt}}l@{\extracolsep{\fill}}c@{}}
				#1 & $\maximize{#2}$ & $#3$ & $ $ \\[5pt]
					 & $\subject\ $    & $#4$ & $ $ \\
					 &     & $#5$ & $ $ \\
					 &     & $#6$ & $ $ \\
					 &     & $#7$ & $ $
			\end{tabular*}}
			}

	\def\cp2problem#1#2#3#4{\fbox
		 {\begin{tabular*}{0.9\textwidth}
			{@{}l@{\extracolsep{\fill}}l@{\extracolsep{6pt}}l@{\extracolsep{\fill}}c@{}}
				#1 & & $#4 $ 
			\end{tabular*}}}

		\renewcommand{\emph}[1]{\textbf{#1}}

		\def\bkE{{\rm I\kern-.17em E}}
		\def\bk1{{\rm 1\kern-.17em l}}
		\def\bkD{{\rm I\kern-.17em D}}
		\def\bkR{{\rm I\kern-.17em R}}
		\def\bkP{{\rm I\kern-.17em P}}
		
		\def\bkZ{{\bf{Z}}}

\newcommand {\beeq}[1]{\begin{equation}\label{#1}}
\newcommand {\eeeq}{\end{equation}}
\newcommand {\bea}{\begin{eqnarray}}
\newcommand {\eea}{\end{eqnarray}}

\def\texitem#1{\par\smallskip\noindent\hangindent 25pt
               \hbox to 25pt {\hss #1 ~}\ignorespaces}

\usepackage{amsmath}
\usepackage{tikz}
\usepackage{caption}
\usepackage{subcaption}
\usepackage{graphicx}
\usepackage{mathtools}
\usepackage{enumitem}
\usepackage{comment}
\usepackage{algorithmic} % <- Preamble
\Crefname{ALC@unique}{Line}{Lines} % <- Preamble

\def\Stilde{\widetilde S}
\newtheorem{example}{Example}
\newcommand{\Rnum}[1]{\uppercase\expandafter{\romannumeral #1\relax}}

\ifpdf
\hypersetup{
  pdftitle={On Independent Cliques and  Linear Complementarity Problems},
  pdfauthor={K. N. Chadha and A. A. Kulkarni}
}
\fi

\myexternaldocument{ex_supplement}

\begin{document}

\maketitle

% REQUIRED
\begin{abstract}
  In recent work (Pandit and Kulkarni [Discrete Applied Mathematics, 244 (2018), pp. 155--169]), the independence number of a graph was characterized as the maximum of the $\ell_1$ norm of solutions of a Linear Complementarity Problem (\LCP) defined suitably using parameters of the graph. Solutions of this LCP have another relation, namely, that they corresponded to Nash equilibria of  a public goods game. Motivated by this, we consider a perturbation of this LCP and identify the combinatorial structures on the graph that correspond to the maximum $\ell_1$ norm of solutions of the new LCP. We introduce a new concept called independent clique solutions which are solutions of the LCP that are supported on independent cliques and show that for small perturbations, such solutions attain the maximum $\ell_1$ norm amongst all solutions of the new LCP.
\end{abstract}

% REQUIRED

% REQUIRED

\section{Introduction}

An undirected graph $G$ is defined as $G = (V,E)$, where $V$ is a finite set of vertices and $E$ is a set of unordered pairs of vertices called {\it edges}. Two vertices $i,j \in V$ are said to be \textit{adjacent} to each other if they have an edge between them, i.e., $(i,j) \in E$.  The \textit{neighborhood} of a vertex is the set of all vertices adjacent to it. A set of vertices such that none of them is adjacent to each other is called an {\it independent set}. An independent set of maximum cardinality is called a \textit{maximum independent set} and its cardinality is denoted as $\alpha(G)$, called the independence number of $G$. An independent set is said to be \textit{maximal} if it not a strict subset of another independent set. If weights $w = (w_1,w_2,\dots,w_{|V|}) \geq 0$ are assigned to each vertex, then a weighted maximum independent set is an independent set which maximizes $\sum_{i \in S}w_i$ over all independent sets $S$ and the value of this maximum is denoted by $\alpha_w(G)$, called the $w$-weighted independence number of $G$.

This paper is about a relation between graphs and a class of continuous optimization problems called Linear Complementarity Problems (LCPs). Given a matrix $M \in \Real^{n \times n}$ and a vector $q \in \Real^n$, LCP($M,q$) is the following problem:
\begin{equation*}
{\rm Find } \ x \in \Real^n \  {\rm such} \ {\rm that} \ \ x \geq 0, \ y = Mx + q \geq 0, \ y\t x =  0.
\end{equation*}
Any such $x$ is called a solution of LCP($M,q$). Our work is motivated by previous results in \cite{pandit2018linear} which established an LCP based characterization of $\alpha_w(G)$. For a graph $G$  on $n$ vertices, let $A \in \{0,1\}^{n\times n}$ denote its adjacency matrix: \ie, $A(i,j)=1$ if $(i,j)\in E$, else it is $0.$
Now consider the LCP given by LCP($I+A,-\bfe$) $=: \LCP(G)$ where $I$ is the $n \times n$ identity matrix, $A$ is the adjacency matrix of $G$ and $\bfe$ is a vector of $1$'s.  The authors of~\cite{pandit2018linear} show that the maximum weighted $\ell_1$ norm amongst the solutions of LCP$(G)$ equals the weighted independence number of $G$. Formally, 
\begin{theorem}\label{thm:parthe1}
 For any simple graph $G$ on  $n$ vertices and a vector of weights $w \geq 0$, we have
 \[
 \alpha_w(G) = \max\{w\t x | x \ {\rm  is \ a \ solution \ of \ \textup{LCP}}(I+A,-\bfe)\}.
 \]
\end{theorem}
It easy to argue that the characteristic vector of any maximal independent set in $G$ solves $\LCP(G)$. Consequently, one trivially has that $\alpha_w(G) \leq w\t x$ for all $x$ that solve $\LCP(G)$. The above result is nontrivial because it shows that no fractional solution $x$ of $\LCP(G)$ can attain a strictly greater value for $w\t x$ than that attained by the characteristic vector of a $w$-weighted maximum independent set.

Indeed there is a third connection which concerns the interpretation of solutions of LCP$(G)$ in terms of Nash equilibria of a {\it public goods game} defined on a network $G$. In the model considered in~\cite{bramoulle2007public}, each vertex $i$ of the graph $G$ is an agent who exerts a scalar effort $x_i \geq 0$ and its utility for the effort profile $x = (x_1,x_2,\dots,x_{|V|})$ is 
\begin{equation*}
 u_i(x) = b\Bigg(x_i + \sum_{j \in N_G(i)}x_j\Bigg) - cx_i,
\end{equation*}
where $N_G(i)$ denotes the neighborhood of vertex $i$ in graph $G$, $b(\cdot)$ is a differentiable strictly concave increasing benefit function and $c$ is the marginal cost of exerting unit effort. Thus each agent benefits from its own effort and the effort exerted by its neighbors (more details on this model can be found in~\cite{bramoulle2007public}). Assume for simplicity that $b^{\prime-1}(c) =1$, \ie, marginal benefit equals marginal cost at unit effort. It is shown in \cite{pandit2018refinement} that the Nash equilibria of the above public goods game are given by  solutions of LCP$(G)$. Moreover, each \textit{maximal independent set} corresponds to an equilibrium of the above game wherein each vertex in the maximal independent set exerts unit effort and all other vertices exert no effort.  It follows that for a given vector of non-negative weights $w$, the maximum weighted effort amongst all equilibria, which is the maximum weighted $\ell_1$ norm amongst the solutions of LCP$(G)$, is achieved by the equilibrium corresponding to the  $w$-weighted maximum independent set. These results illustrate how Nash equilibria of this game are intimately related to combinatorial structures on the underlying graph.

In this paper, we are concerned with perturbations of the above model. Specifically, we ask, what happens if the argument of the benefit function is $x_i + \delta \sum_{j \in N_G(i)}x_j$ instead of $x_i + \sum_{j \in N_G(i)}x_j$? Here $\delta>0$ is a substitutability factor which captures the case when the benefit agent $i$ derives from its neighbors is proportional to $\delta$ times the sum of efforts exerted by its neighbors. We view $\delta$ as a small perturbation from unity, \ie, $\delta$ is close to unity, but may not be exactly unity. Thus we think of $\delta=1$ as the idealized case where efforts of neighbors substitute exactly for the agent's own effort, and  $\delta \neq 1$ can be thought of as arising due to small losses or misspecifications from this idealized situation. 
The existence of equilibria in this case has been studied in \cite{bramoulle2014strategic}. It is easy to argue that an equilibrium exists and all the equilibria of this new game are given by solutions of LCP$(I + \delta A,-\bfe) =:\LCP_\delta(G)$. With this motivation, in this paper, we characterize $\ell_1$ norm maximizing solutions of $\LCP_\delta(G)$ for $\delta \neq 1$, but close to $1$. 

Our effort is to relate these solutions to combinatorial structures on the underlying graph.
Before we mention our contributions, we make a few observations about $\LCP_\delta(G).$ 
First, we observe that small perturbations $\delta$ around unity have a nontrivial effect on the combinatorial structure of solutions. In particular, for $\delta<1$, it is not true that the characteristic vector of every maximal independent set solves $\LCP_\delta(G)$. In fact, a binary vector solves $\LCP_\delta(G)$ if and only if it is the characteristic vector of a $\lceil \frac{1}{\delta}\rceil$-dominating independent set, an object which may not exist in $G$. In the case when a $\lceil \frac{1}{\delta}\rceil$-dominating independent set does not exist, identifying combinatorial structures that support solutions and maximize the $\ell_1$ norm becomes challenging and forces one to expand the search space of combinatorial structures that can be identified with solutions of $\LCP_\delta(G)$.

In this paper we do precisely this. We introduce a new concept called {\it independent cliques solutions} (ICS) which are defined as solutions of $\LCP_\delta(G)$  whose support is a union of independent cliques.  Two cliques in a graph are said to be \textit{independent} if no vertex of one clique has any vertex of the other clique as its neighbor. Independent cliques can be thought of as a generalization of independent sets since when each clique is degenerate (a single vertex), a union of independent cliques is an independent set. We prove that the maximum $\ell_1$ norm amongst all solutions of $\LCP_\delta(G)$ is achieved by an ICS  with $\alpha(G)$ cliques for $\delta \in [\eta(G),1)$, where 
\begin{equation*}
 \eta(G) =  \max\Big{\{}\frac{\omega(G)-3 + \sqrt{(\omega(G)-3)^2 + 4(\omega(G)-1)}}{2(\omega(G)-1)},\frac{\alpha(G)(\omega(G)-1) - \omega(G)}{\alpha(G)(\omega(G)-1)}\Big{\}},
\end{equation*}
and $\omega(G)$ is the size of the largest clique in $G$. Thus while $\ell_1$ norm maximizing solutions of $\LCP(G)$ include those supported by a maximum independent set, \ie, $\alpha(G)$ degenerate cliques, for $\delta \in [\eta(G),1)$ the corresponding solutions of $\LCP_\delta(G)$ comprise of $\alpha(G)$ not-necessarily-degenerate cliques.  
Moreover, in the case when when a unique maximum independent set exists in a graph $G$ the characteristic vector of the unique maximum independent set is a solution of $\LCP_\delta(G)$ and is, in fact, its $\ell_1$ norm maximizing solution for $\delta \in [\eta(G),1)$. Lastly, we show that for $\delta \geq 1$, the results of \cite{pandit2018linear} continue to hold, i.e., the maximum weighted $\ell_1$ norm amongst the solutions of $\LCP_\delta(G)$ is the weighted independence number achieved by the characteristic vector of a $w$-weighted maximum independent set.

These results are proved as follows. We show the existence of ICSs via an algorithm (\cref{alg:ics}) that constructs an ICS for any graph. We prove that \cref{alg:ics} outputs a vector, as a function of $\delta$ with support as a union of independent cliques. Moreover, the support does not depend on $\delta$. From this we show that these independent cliques support an ICS for all $\delta \in [\gamma(G),1)$, where
\begin{equation*}
 \gamma(G) = \frac{\omega(G)-3 + \sqrt{(\omega(G)-3)^2 + 4(\omega(G)-1)}}{2(\omega(G)-1)}.
\end{equation*}
We also prove that the lower bound on $\delta$, namely $\gamma(G)$, is tight by showing via examples that when this $\delta < \gamma(G)$, an ICS need not exist. Next, we prove that the ICS of $\LCP_\delta(G)$ which achieves the maximum $\ell_1$ norm amongst all the ICSs also achieves the maximum $\ell_1$ norm amongst all solutions of $\LCP_\delta(G)$ for $\delta \in [\eta(G),1)$.

The problem of characterizing $\ell_1$ norm maximizing solutions of $\LCP_\delta(G)$ is highly complex since the solution set of the LCP is not convex (it is a union of polyhedra~\cite{cottle92linear}) and no known graph structures directly provide solutions to $\LCP_\delta(G)$. Our results show that for $\delta \geq \eta(G)$, an ICS always exists and is also $\ell_1$ norm maximizing amongst all solutions of $\LCP_\delta(G)$. It is also easy to show that for $\delta < -\frac{1}{\lambda_{\min}(A)}$, $\LCP_\delta(G)$ admits a \textit{unique} solution that is related to centrality notions on graphs (see~\cite{bramoulle2014strategic}). It would be fascinating to ascertain the combinatorial structures that characterize $\ell_1$ norm maximizing solutions of $\LCP_\delta(G)$ for the entire range of $\delta$ from $-\frac{1}{\lambda_{\min}(A)}$ to unity. Though our results do not span this range, we believe they nonetheless provide interesting relations between the structural properties of graphs and solutions of $\LCP_\delta(G)$. 

\subsection{Related Work} Our results, in effect, give a characterization of the solutions to a special class of Linear Programs with Complementary Constraints (LPCC). In its most general form, an LPCC is defined as 

\maxfourproblem{LPCC}{x,y}{c\t x + d\t y}{Bx + Cy \geq b}{Mx + Ny + q \geq 0}{x \geq 0}{x\t(Mx + Ny + q) = 0 } \\

When we take $B,C,b,N,d$ to be $0$ in the LPCC, and take $c = \bfe$, $M = I + \delta A$ and $q = -\bfe$, the LPCC reduces to the problem we consider. LPCCs provide a generalization to problem classes such as linear programming and finding sparse (minimum $\ell_0$ norm) solutions of linear equations \cite{hu12lpcc, hu12linear}. The results in \cite{pandit2018linear} show that it is hard to find approximate solutions of an LPCC. While the LPCC is a newly explored topic, LCPs are deeply studied subjects, book-length treatments of which can be found in \cite{cottle92linear} and \cite{murty1988linear}. 

The idea that Nash equilibria of games can be related to solutions of LCPs is not new. Consider a simultaneous move game with two players (\Rnum{1},\Rnum{2}), where player \Rnum{1} has $m$ possible actions and the player \Rnum{2} has $n$ possible actions. The cost matrices $A_{\rm \Rnum{1}},A_{\rm \Rnum{2}} \in \Real^{m \times n}$ are such that when player \Rnum{1} chooses action $i$ and player {\Rnum{2}} chooses action $j$, they incur costs $A_{\rm \Rnum{1}}(i,j)$ and $A_{\rm \Rnum{2}}(i,j)$ respectively. Players can also choose to play mixed strategies which are vectors defined over the probability simplex in a $m$ dimensional space for player \Rnum{1} and $n$ dimensional space for player \Rnum{2}. A Nash equilibrium in mixed strategies in this game is defined as a pair of vectors $x^* \in \Delta^n$, $y^* \in \Delta^m$ such that 
\[
 (x^*)\t A_{\rm \Rnum{1}}y^* \leq x\t A_{\rm \Rnum{1}}y^* \ \forall x \in \Delta^n \ \text{ and } (x^*)\t A_{\rm \Rnum{2}}y^* \leq (x^*)\t A_{\rm \Rnum{2}}y \ \forall y \in \Delta^m,
\]
where $\Delta^k$ is a probabilty simplex in $\Real^k$, $\Delta^k \coloneqq \{x \in \Real^k | \sum_ix_i = 1, x_i \geq 0\}$. Assuming $A_{\rm \Rnum{1}}$ and $A_{\rm \Rnum{2}}$ are entrywise positive matrices, we define ($\tilde x , \tilde y $) as
\[
 \tilde x = \frac{x^*}{(x^*)\t A_{\rm \Rnum{2}}y^*} \text{ and } \tilde y = \frac{y^*}{(x^*)\t A_{\rm \Rnum{1}}y^*}.
\]
It can be shown (\cite{cottle92linear}) that ($\tilde x , \tilde y $) satisfy LCP($M,q$) with 
\[
M = 
 \begin{bmatrix}
  0 & A_{\rm \Rnum{1}} \\
  A_{\rm \Rnum{2}}\t  & 0 \\
 \end{bmatrix}
  \text{ and } q = -\bfe.
\]
More generally, certain equilibria of games involving coupled constraints \cite{kulkarni2012variational} also reduce to LCPs.
 
There has been prior effort at relating LCPs with independent sets (\cite{locatelli2004combinatorics}, \cite{massaro2002complementary}). Particularly, the authors in \cite{massaro2002complementary} show that $x$  is a characterestic vector of a maximal independent set  if and only if $(\frac{x}{||x||_1},y_1,y_2)$ solves LCP($M_G,q_G$) with 
\[
M = 
 \begin{bmatrix}
  A + I & - \bfe &  \bfe \\
  \bfe \t & 0 &  0 \\
  -\bfe \t & 0 &  0 \\
  \end{bmatrix}
  \text{ and } q = \begin{bmatrix}
  0 \\
  \vdots \\
  0 \\
  -1 \\
  1 \\
  \end{bmatrix}.
\]
On the contrary, our work is to characterize $\ell_1$ norm maximizing solutions of LCP$_\delta(G)$, solutions to which lie in a different space ($\Real^n$) compared to those of  LCP($M_G,q_G$) ($\Real^{n+2}$). It is shown in \cite{pandit2018linear} that the characterestic vectors of maximal independent sets are solutions of LCP$(G)$, but the same is not the case with solutions of LCP$_\delta(G)$ for general $\delta$. Our work is distinct from both \cite{pandit2018linear} and \cite{massaro2002complementary} since we consider a different class of LCPs and relate their equilibria to combinatorial structures in the graph.

Generalizations of independent sets have been studied in other contexts. The authors of \cite{siemes1994unique} generalize independent sets to $k$-independent sets where a set $I$ of vertices of $G$ is said to be $k-$independent if $I$ is independent and every independent subset $I$ of $G$ with $| I| \geq | I| -(k - 1)$ is a subset of $I$. It is easy to note that $0$-independent sets are independent sets and $1$-independent sets are unique maximum independent sets. This generalization is more restrictive, and $k$-independent sets need not even exist for all graphs for all $k \geq 1$. In contrast, we provide a more inclusive generalization to a union of independent cliques which includes all independent sets as a special case. To the best of our knowledge, such a generalization is the first of its kind. We also study the special case when unique maximum independent sets exist. These sets need not always exist for a graph. The authors of \cite{hopkins1985graphs} provide sufficient conditions on the graph under which such independent sets exist.

\subsection{Organization of the paper} The rest of the paper is organized as follows. In \cref{sec:pre}, we provide the preliminaries and the notation used throughout the paper along with an introduction to LCPs. In \cref{sec:lcpg}, we discuss properties of LCP$_\delta(G)$ and its solutions. In \cref{sec:les1}, we introduce the notion of Independent Clique Solutions (ICS), provide an algorithm to find them and prove that they achieve the maximum $\ell_1$ norm amongst all LCP$_\delta(G)$ solutions for $\delta < 1$. In \cref{sec:geq1}, we extend the results of \cite{pandit2018linear} to the case when $\delta \geq 1$. The paper concludes in \cref{sec:con}.

%%%%%%%%%%%%%%%%%%%%%%%%%%%%%%%%%%%%%%%%%%%%%%%%%%%%%%%%%%%%%%%%%%%%%%
%%%%%%%%%%%%%%%%%%%%%%%%%%%%%%%%%%%%%%%%%%%%%%%%%%%%%%%%%%%%%%%%%%%%%%

\section{Preliminaries and Notation}
\label{sec:pre}

For a graph $G$, we denote by $V(G)$ and $E(G)$ its vertex and edge sets respectively. For a vertex $i \in V(G)$, let $N_G(i)$ denote the neighborhood of $i$ in $G$, i.e. $N_G(i) = \{j \in V(G)|(i,j) \in E(G)\}$. For a set of vertices $K \subset V(G)$,   $N_G(K) = \bigcup_{i \in K}N_G(i) \backslash K$. Also, let $\Nbar_G(K) = N_G(K) \cup K$ denote the closed neighborhood of $K$ in $G$. For $S \subset V(G)$, let $G_S$ denote the graph restricted to the vertex set $S$ and for a vector $x$ indexed by $V(G)$, let $x_S$ denote the subvector of $x$ with components indexed by the vertex set $S$. For a vector $x$ indexed by $V(G)$, let $\sigma(x)$ be its support, i.e. 
\[
\sigma(x) \coloneqq \{i \in V(G)| x_i > 0\}.
\]
Let $\bfone_S$ denote the characteristic vector of set $S$, i.e.
\begin{equation*}
\bfone_S(i) \coloneqq (\bfone_S)_i =  
\begin{cases}
1 , & {\rm  if } \ i \in S \\
0 , & {\rm  otherwise }.
\end{cases}
\end{equation*}
We use the standard notation $||\cdot||_1$ for the $\ell_1$ norm. 
 Let $A$ denote the adjacency matrix of a graph $G$ given by
\begin{equation*}
A(i,j) = a_{ij} =  
\begin{cases}
1 , & {\rm  if } \ (i,j) \in E(G) \\
0 , & {\rm  otherwise }
\end{cases}
\end{equation*}

An \textit{independent set} of a graph is defined as a set of vertices such that none of them are neighbors. A \textit{maximal independent set} is defined as an independent set $S$ such that all the vertices which are not in $S$ have at least one neighbor in $S$, i.e. $|N_G(i) \cap S| \geq 1 \ \forall i \not\in S$. A \textit{maximum} independent set is an independent set which has the highest cardinality amongst all independent sets. The cardinality of the maximum indepdendent set is denoted by $\alpha(G)$ and is also called the independence number of the graph $G$. Let $\alpha_w(G)$ denote the weighted independence number for $w \geq 0$ which is the maximum sum of weights amongst all indepdendent sets, i.e. $\alpha_w(G) \coloneqq \max\{\sum_{i \in S}w_i | S \subset V(G) \ {\rm  independent }\}$. A \textit{dominating} set ($D$) is a set such that $N_G(D) \cup D = V(G)$. Thus, an independent set which is also dominating is a maximal independent set. A \textit{$k$-dominating} set is defined as a dominating set $D$ such that $|N_G(i) \cap D| \geq k \ \forall i \in V(G) \backslash D$. A set which is both $k$-dominating and independent is called a $k$-dominating indepdendent set. It can be observed that a $k$-dominating independent set is also a $m$-dominating independent set for all $m < k$. For $k  = 1$, these are equivalent to maximal independent sets and they always exist. For $k\geq 2$, they may or may not exist depending on the graph. Examples of both the cases is shown in \cref{fig:somisex}.

\begin{figure}[H]
\centering
\begin{subfigure}{.5\textwidth}
  \centering
  \includegraphics[scale = 1]{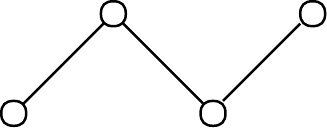}
  \caption{No $2$-dominating independent sets}
  \label{fig:nosomis}
\end{subfigure}%
\begin{subfigure}{.5\textwidth}
  \centering
  \includegraphics[scale = 1]{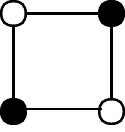}
  \caption{The shaded vertices form a $2$-dominating independent set}
  \label{fig:somis}
\end{subfigure}
\caption{Examples of existence and non-existence of $2$-dominating independent sets}
\label{fig:somisex}
\end{figure}

\section{LCP$_\delta(G)$}\label{sec:lcpg}

For a graph $G$ with adjacency matrix $A$ recall the problem $\LCP(I+\delta A,-\bfe)=\LCP_\delta(G)$ where $\delta$ is a positive parameter. 
%We establish some basic properties of 
%\[
%\boxed{\textup{LCP}_\delta(G) \ {\rm Find} \ x \in \Real^n \ {\rm such \ that} \ x \geq 0, (I + \delta A)x \geq \bfe, x\t ((I+\delta A)x - \bfe) = 0}
%\]
%
In this section we prove some properties of LCP$_\delta(G$) that  will be used later in the paper. 

Define $\C(x) \coloneqq (I + \delta A)x$ and denote by $\C_i(x)$ the $i\th$ component of $\C(x)$. $\C_i(x)$ is called as the \textit{discounted sum of the closed neighborhood} of $i$ with respect to $x$,
\begin{equation}\label{eqn:nbdsum}
\C_i(x) = x_i + \delta \sum_{j \in V(G)}a_{ij}x_j =  x_i + \delta \sum_{j \in N_G(i)}x_j.
\end{equation}
We denote the set of solutions of LCP$_\delta(G$) by SOL$_\delta(G$). Clearly $x \in {\rm SOL}_\delta(G)$ if and only if,
\begin{align}\label{eqn:lcp1}
 & x_i \geq 0,  \\ \label{eqn:lcp2}
 & \C_i(x) \geq 1, \\ \label{eqn:lcp3}
\text{and } & x_i(\C_i(x) - 1) = 0 
\end{align}
 $\forall i \in V(G)$.  
We first show a lemma that establishes some basic properties of SOL$_\delta(G$).
\begin{lemma}\label{lem:gen}
Consider the $\LCP_\delta(G) = \LCP(I+\delta A, -e)$. Then,
\begin{enumerate}[label = (\alph*)]
\item $0 \not \in$ $\SOL_\delta(G$),
\item $C(x) \geq x \ \forall x \in \SOL_\delta(G)$,
\item $\SOL_\delta(G) \subset [0,1]^n$,
\item If a graph G is a disjoint union of graphs $G_1$ and $G_2$, then $\textup{SOL}_\delta(G) = \textup{SOL}_\delta(G_1) \times \textup{SOL}_\delta(G_2)$,
\item For a graph $G$, if $x \in \textup{SOL}_\delta(G)$, $\sigma(x)$ is a $\lceil \frac{1}{\delta} \rceil $-dominating set of $G$,
\item For a graph $G$, if $x \in \textup{SOL}_\delta(G)$, $\tilde{x} = x_{\sigma(x)} \in \textup{SOL}_\delta(G_{\sigma(x)})$ and $\sigma{(\tilde{x})} = V(G_{\sigma(x)})$.
\end{enumerate}
\end{lemma}

\begin{proof}
See \cref{appen:lem:gen}.
\end{proof}

The following lemma characterizes integer (binary) solutions of $\LCP_\delta(G)$.

\begin{lemma}\label{lem:intsol}
For a graph $G$ and $\delta \in (0,\infty)$, $x$ is an integer solution of \textup{LCP}$_\delta(G)$ if and only if it is a characteristic vector of a $\lceil \frac{1}{\delta} \rceil $-dominating independent set of $G$
\end{lemma}

\begin{proof}
Let $x$ be an integer solution of LCP$_\delta(G)$. Then, by \cref{lem:gen} (c), $x$ is a binary vector, and hence $x = \bfone_S$ for a subset $S$ of vertices in $G$. If $S$ is not an independent set, then there exist $i,j \in S$ such that $j \in N_G(i)$, i.e. $x_i = x_j = 1$, whereby $\C_i(x) = x_i + \delta \sum_{j \in N_G(i)}x_j \geq 1+\delta > 1$. Thus, $\C_i(x) - 1 > 0$ and $x_i > 0$ imply \cref{eqn:lcp3} is violated. This gives a contradiction. Hence, $S$ must be an independent set. From \cref{lem:gen} (e) we have that $S$ must also be a $\lceil \frac{1}{\delta} \rceil $-dominating set, as required.

For the converse, let $S$ be a $\lceil \frac{1}{\delta} \rceil $-dominating indepdendent set of $G$. Then, for $i \in S$, $\C_i(x) = 1 + \delta \sum_{j \in N_G(i)}x_j = 1$ since $x_j = 0 \ \forall j \in N_G(i)$. Also, for $i \not\in S$, $\C_i(x) = 0 + \delta \sum_{j \in N_G(i)}x_j \geq \delta \lceil \frac{1}{\delta} \rceil \geq \delta \times \frac{1}{\delta} = 1$. Thus, we have $\C_i(x) \geq 1$ and $x_i(\C_i(x) - 1) = 0$ for all $i \in V(G)$. Hence, if $x$ is a characteristic vector of a $\lceil \frac{1}{\delta} \rceil $-dominating indepdendent set of $G$, then it is an integer solution of LCP$_\delta(G)$.
\end{proof}

Consider the following optimization problem which gives the $\ell_1$ norm maximizing solution amongst SOL$_\delta(G)$: 
\begin{center}
\maxproblem{{\rm {\rm maxSOL}}$_\delta(G)$}{}{\bfe\t x}{x \in {\rm SOL}_\delta(G).} 
\end{center}
We next show that if a $\lceil \frac{1}{\delta} \rceil $-dominating independent set $S$ is contained in the support of a solution $x$ then $\bfone_S$ has  $\ell_1$-norm no lesser than that of $x$.
\begin{lemma}\label{lem:card}
For a graph $G$ and $\delta \in (0,\infty)$, if $x$ is a solution of \textup{LCP}$_\delta(G)$ and $S \subset \sigma(x)$ is a $\lceil \frac{1}{\delta} \rceil $-dominating independent set of $G$, then $||x||_1 \leq |S|$.
\end{lemma}

\begin{proof}
Let $U = \sigma(x) \backslash S$. From \cref{eqn:lcp3}, we know that
\begin{equation}\label{eqn:card_bas}
\C_i(x) =  x_i + \delta \sum_{j \in V(G)}a_{ij}x_j = 1 \ \forall i \in \sigma(x). 
\end{equation}

Summing over $i \in \sigma(x)$, we get
\begin{equation}\label{eqn:card}
\sum_{i \in V(G)}x_i + \delta \sum_{j \in \sigma(x)}\sum_{j \in V(G)}a_{ij}x_j = |\sigma(x)| = |S| + |U|.
\end{equation}

On rearranging, 
\begin{align}
\nonumber |S| - \bfe\t x &= \delta \sum_{i \in \sigma(x)}\sum_{j \in V(G)}a_{ij}x_j - |U| \\
&= \delta \sum_{i \in S}\sum_{j \in U}a_{ij}x_j + \delta \sum_{i \in U}\sum_{j \in V(G)}a_{ij}x_j - |U| \label{eqn:card2}\\
&= \delta \sum_{j \in U}\sum_{i \in S}a_{ij}x_j + \sum_{i \in U}(\delta \sum_{j \in V(G)}a_{ij}x_j - 1) \label{eqn:card3}\\
&= \delta \sum_{j \in U}N_S(j)x_j - \sum_{i \in U}x_i \label{eqn:card4}\\
&\geq  \sum_{j \in U}x_j - \sum_{i \in U}x_i = 0 \label{eqn:card5}
\end{align}
The equality in \cref{eqn:card2} follows from the fact that $a_{ij} = 0 \ \forall i,j \in S$ and $x_j = 0 \ \forall j \not\in \sigma(x)$, \cref{eqn:card3} follows from \cref{eqn:card_bas}. The inequality \cref{eqn:card5} follows since $S$ is a $\lceil \frac{1}{\delta} \rceil $-dominating independent set of $G$    
\end{proof}
Consequently, if $G$ is such that every solution $x$ of $\LCP_\delta(G)$ contains a $\lceil \frac{1}{\delta} \rceil $-dominating independent in its support, then the solution of maxSOL$_\delta(G)$ would attained a characteristic vector of a $\lceil \frac{1}{\delta} \rceil $-dominating independent set. Clearly, such a property does not hold for all graphs $G$ since $\lceil \frac{1}{\delta} \rceil $-dominating independent sets do not exist in every graph. This requires us to analyze the solutions of maxSOL$_\delta(G)$ more carefully.

%\subsection{Potential function}\label{sec:potfun}

We now introduce a function called the \textit{potential function}, the stationary points of which are the solutions of LCP$_\delta(G)$. This was proved by \cite{bramoulle2014strategic} in the context of games on networks and is known more generally in the LCP literature~\cite{cottle92linear}. Define,
\begin{equation}\label{eqn:potfun}
\phi(x;\delta,G) \coloneqq x\t \bfe - \frac{1}{2}x\t (I+\delta A)x \qquad \forall x \in \Real^n
\end{equation}
where $n=|V|.$
\begin{lemma}\label{lem:lcp-potfun}
$\SOL_\delta(G)$ equals the set of stationary points of $\phi(x;\delta,G)$ in ${\Real^n_+}$.
\end{lemma}

\begin{proof}
We find the Karush Kahn Tucker (KKT) conditions of the following maximization problem:
\[
\max \ \phi(x;\delta,G) \ s.t. \ x_i \geq 0. 
\] 
We get $x$ is a stationary point iff $\exists \ \mu \geq 0$ such that $\bfe - (I + \delta A)x = -\mu$ and $x\t \mu = 0$. These conditions are equivalent to the LCP$_\delta(G)$ conditions \cref{eqn:lcp1,eqn:lcp2,eqn:lcp3}.
\end{proof}

We note that, when $\delta < -\frac{1}{\lambda_{\min}(A)}$, where $\lambda_{\min}(A)$ represents the minimum eigenvalue of the adjacency matrix of the graph $G$, the potential function is a strictly concave function. Thus, in this case, SOL$_\delta(G)$ is a singleton.

\section{Solutions of LCP$_\delta(G)$ for $\delta < 1$} \label{sec:les1}

In this section, we assume throughout that $\delta<1$. Our goal is to prove that the solution of {\rm maxSOL}$_\delta(G)$ is achieved by a member of the class of solutions that we call independent clique solution (ICS). To prove this, we first formally define independent clique solutions and prove their existence constructively (\cref{sec:ind_cliq_def}). We also prove some results relating this special class of solutions to independent sets. We then show that the maximum $\ell_1$ norm amongst the solutions in this class is monotonically increasing under graph inclusion (\cref{subsec:monotonic}). Using this, we inductively argue that the solution of {\rm maxSOL}$_\delta(G)$ is achieved by an independent clique solution (\cref{subsec:proof}).

\subsection{Independent Clique Solutions}\label{sec:ind_cliq_def}

First, we note the following definition:

\begin{definition}
Two cliques in a graph are said to be \textit{independent} if no vertex of one clique has any vertex of the other clique as its neighbor.
\end{definition}  

Note that the above definition is a generalization of the definition of independence of vertices which are in fact $K_1$'s.

\begin{definition}
Let $G$ be a graph and $\delta>0$. An \textit{independent clique solution (ICS)} is a solution of $\LCP_\delta(G)$ whose support is a union of independent cliques.
\end{definition}  

We denote the set of independent clique solutions of LCP$_\delta(G)$ as ICS$_\delta(G)$. Consider the following optimization problem which gives us the $\ell_1$ norm maximizing solution amongst ICS$_\delta(G)$: \\

\maxproblem{{\rm maxICS}$_\delta(G)$}{}{\bfe\t x}{x \in \textup{ICS}_\delta(G)} \\

\begin{comment}
Note that, since $\delta < 1$, the support of an ICS must always be a $2$-dominating set. 
If not, consider a vertex (say $i$) which has only one neighbor in the support of the solution of LCP$_\delta(G)$. Then, $\C_i(x) = x_i + \delta \sum_{j \in N_G(i)} x_j \leq \delta < 1$, since $x_i = 0,$ and $x_j < 1 \ \forall j \in V(G)$. This violates \cref{eqn:lcp2}. Thus, every vertex which is not in the support of an ICS must have at least two neighbors in the support of the ICS. More generally, it can be said that the support must be $\lceil \frac{1}{\delta} \rceil$-dominating. Since $\delta < 1$, it is definitely $2$-dominating.
\end{comment}

We now show that for any graph $G$ an ICS exists if $\delta$ is greater than a threshold $\gamma(G) \in (0,1)$. We give a constructive proof. \cref{alg:ics} gives a method to construct an ICS using maximum independent sets in any graph. The justification of why this algorithm gives an ICS under suitable conditions on $\delta$ is given in \cref{thm:cliqexist}. 

\begin{algorithm}
\caption{Construct an Independent Clique solution}
\label{alg:ics}
\hspace*{\algorithmicindent} \textbf{Input:} $G(V,E)$ \\
\hspace*{\algorithmicindent} \textbf{Output:} $x$
\begin{algorithmic}[1]
\STATE{Initialize $\Vbar = V(G)$, $\Gbar = G$, $\Sbar = S = $ a maximum independent set of $\Gbar$ and $x = \bfone_S$}\label{line:init1}
\STATE{Set {$\Sbar = \{1,2,\cdots,{|S|}\}$} and define $L = \{l \in V| \ l {\rm \ has \ only \ one \ neighbor \ in} \ S\}$}\label{line:Ldef}
\FOR{$i \in \{1,2,\cdots,|S|\}$}\label{line:bigfor}
\STATE{Define $C_{i} = N_{\Gbar}(i) \cap L$, $\Cbar_{i} = C_{i} \cup \{i\}$ and $Q_{i} = N_{\Gbar}(\Cbar_{i})$}\label{line:C_def}
\FORALL{$j \in \Cbar_{i}$} \label{line:for}
    \STATE{$x_j = \frac{1}{1+(|\Cbar_{i}|-1)\delta}$} \label{line:update}
\ENDFOR\label{line:endfor}
\STATE{Define $ V' \coloneqq \Vbar \backslash (\Nbar_{\Gbar}(\Cbar_{i}))$}\label{line:V'def}
\STATE{Define $ S' \coloneqq \Sbar \backslash \{i\}$}\label{line:S'def}
\STATE{Set $\Vbar \leftarrow V'$, $\Gbar \leftarrow \Gbar_{V'}$}\label{line:Vdef}
\STATE{Set $\Sbar \leftarrow S'$}\label{line:Sdef}
\ENDFOR\label{line:endbigfor}
\RETURN $x$
\end{algorithmic}
\end{algorithm}

We will show that this algorithm returns $x$ which is an ICS supported on $\Cbar_1\cup \Cbar_2 \cdots \cup \Cbar_{|S|}$ for a maximum independent set $S$ in $G$. Below we argue in \cref{lem:complete} that these sets $\{\Cbar_i\}$ are independent cliques. In Theorem~\ref{thm:cliqexist}, we show that the $x$ returned solves the LCP. We first note a few remarks.
\begin{remark} \label{rem:algremS} 
In \cref{alg:ics}, {for each $i$}, $\Nbar_{\Gbar}(\Cbar_{i}) \cap S = \{i\}$. To see this, note that $\Nbar_{\Gbar}(\Cbar_{i}) = \{i\} \cup N_{\Gbar}(i) \cup N_{\Gbar}(C_{i})$. Since $i \in S$, $N_{\Gbar}(i) \cap S = \emptyset$. Since $C_{i} \subset L$,  $N_{\Gbar}(C_i) \cap S = \emptyset$. Thus, no vertex in $\Nbar_{\Gbar}(\Cbar_{i}) \backslash \{i\}$ is in $S$. 
\end{remark}

\begin{remark}\label{rem:algremV}
In \cref{alg:ics}, at the return step, $\Vbar = \emptyset$, i.e., we eventually remove all vertices from the graph. To see this, the set of removed nodes in each step, say $R_{i}$, is such that $\Nbar_G(k) \subset \bigcup_{i = 1}^{k} R_{i} \subset V$. Taking union on both sides from $k = 1$  to $|S|$, we get $\bigcup_{k = 1}^{|S|} \Nbar_G(k) \subset \bigcup_{i = 1}^{|S|} R_{i} \subset V$. $S$ being a maximal independent set, $\bigcup_{i = 1}^{k}\Nbar_G(k) = V$ and hence $\bigcup_{i = 1}^{|S|} R_{i} = V$.
\end{remark}

Now, we prove a few lemmas about the algorithm that help us prove its validity. First, we show that after each iteration, the updated values $S'$ and $V'$ defined in \cref{line:S'def} and \cref{line:V'def} satisfy the conditions satisfied by $\Sbar$ and $\Vbar$ respectively before the iteration.

\begin{lemma}\label{lem:inductive}
Let $\Vbar$ and $\Sbar$ be updated as in \cref{line:Vdef} and \cref{line:Sdef} in \cref{alg:ics} respectively. Then, after any iteration of the loop \cref{line:bigfor} to \cref{line:endbigfor},
\begin{enumerate}[label = (\alph*)]
\item $\Sbar \subset \Vbar$.\label{lem:inductive1}
\item $\Sbar$ is a maximum independent set of $\Gbar$.\label{lem:inductive2}
\end{enumerate}
\end{lemma}
\begin{proof} 
~\\
\begin{enumerate}[label = (\alph*)]
\item It suffices to show that $S' \subset V'$ with $S'$ and $V'$ as defined in \cref{line:V'def} and \cref{line:S'def}. By definition, we have $S' \subset V$. To show that $S' \subset V'$, we need to show that in the vertices we removed from $V$ to form $V'$, there is no member of $S'$. Suppose we have executed \cref{line:Sdef} of the $k\th$ iteration of the algorithm for some $k \leq |S|$. The set of vertices removed from $V$ at this stage to form $V'$ is $\bigcup_{i = 1}^{k} (\Cbar_{i} \cup Q_{i})$, and we need to show that $\Big(\bigcup_{i = 1}^{k} (\Cbar_{i} \cup Q_{i})\Big) \cap S' = \emptyset$. Now, for each $i$, $(\Cbar_{i} \cup Q_{i})$ contains $i$ which is a member of $S$ but not $S'$. The set $(\Cbar_{i} \cup Q_{i})$ also has neighbors of $i$ which can not be in $S'$ since $S' \subset S$ and $S$ is an independent set. Lastly, $(\Cbar_{i} \cup Q_{i})$ has neighbors of $C_{i}$ (defined on \cref{line:C_def} of \cref{alg:ics}) that are also not in $S'$ since $C_{i}$ is such that $\forall c \in C_{i}, \ N_G(c) \cap S = \{i\}$ and hence $N_G(c) \cap S' = \emptyset$. Since this is true for all $i \in \{1,2,\cdots, k\}$, $\bigcup_{i = 1}^{k} (\Cbar_{i} \cup Q_{i}) \cap S' = \emptyset$, and hence $S' \subset V'$.

\item We show this by induction on the iteration number. At the start of the first iteration of the for loop (\cref{line:bigfor} to \cref{line:endbigfor} in \cref{alg:ics}), $\Sbar$ is a maximum independent set of $\Gbar$. Assume that $\Sbar$ is a maximum independent set of $\Gbar$ after $k-1$ iterations (\ie, at the start of iteration $k$). To show that the claim is true at the start of the $(k+1)\th$ iteration, it suffices to show that $S'$ is a maximum independent set of $\Gbar_{V'}$ at the end of $k\th$ iteration. Note that $|S'| = |\Sbar| - 1$, since only one element ($k$) is removed from $\Sbar$ to form $S'$. Suppose $S'$ is not a maximum independent set in $\Gbar_{V'}$. Then there exists an independent set $S''$ of $\Gbar_{V'}$  such that $|S''| > |S'| = |\Sbar| - 1$ and hence $|S''| \geq |\Sbar|$. Now note that $S' \subset V'$ and since $ V'=\Vbar \backslash (\Nbar_{\Gbar}(\Cbar_{k}))$ there is no edge between any vertex in $V'$ and $\{k\}$. But then $\Shat = S'' \cup \{k\}$ is an independent set of $\Gbar$ with $|\Shat| = |S''| + 1 > |\Sbar|$, which is a contradiction. Thus, $S'$ is a maximum independent set in $\Gbar_{V'}$.
\end{enumerate}
\end{proof}

We now come to the first step of showing that \cref{alg:ics} returns an ICS, namely, showing $\{\Cbar_i\}$ are independent cliques.
\begin{lemma}\label{lem:complete}
$\Cbar_i$, as defined in \cref{line:C_def} of \cref{alg:ics}, forms a clique. Moreover, any two cliques $\Cbar_{i}$ and $\Cbar_{j}$ for $i \neq j$ and $i,j \in S$ are independent.
\end{lemma}

\begin{proof}
First we prove that $\Cbar_{i}$ forms a clique for all $i \in \{1,2,\hdots, |S|\}$.  At the $i\th$ iteration of \cref{alg:ics} let $C_{i}$ be as defined in \cref{line:C_def} and let $\Sbar$ be a maximum independent set of $\Gbar$ with $\Vbar = V(\Gbar)$. To show that $\Cbar_i$ forms a clique, it suffices to show that $C_{i}$ forms a clique since all vertices of $C_{i}$ are neighbors of $i$. We prove this by contradiction. Suppose $C_{i}$ is not a clique. Then choose an independent set $S_{i}$ of $C_{i}$ such that it has at least two vertices. Consider $\Shat = (S \backslash \{i\}) \cup S_{i}$. We claim that $\Shat$ is an independent set. To prove this, it suffices to show that $\forall p \in S \backslash \{i\}, \forall q \in S_{i}, \ (p,q)$ is not an edge in $G$. Now, $S_{i} \subset L$ whereby each vertex in $S_{i}$ has only one neighbor in $S$. Moreover, $S_{i} \subset N_{G}(i)$, whereby for all vertices in $S_{i}$, their only neighbor in $S$ is $i$. Consequently, $\Shat$ is an independent set. Moreover, $|\Shat| > |S|$, which is a contradiction since $S$ must be a maximum independent set of $G$ by definition. Hence, $C_{i}$ and therefore $\Cbar_{i}$ must be a clique. Clearly, this holds for $\forall i \in \{1,2,\cdots, |S|\}$.

Now, consider two cliques $\Cbar_{i}$ and $\Cbar_{j}$ for $i \neq j$. Without loss of generality, we assume $i < j$. Then, $\Cbar_{j} \subset V \backslash \Nbar_G(\Cbar_{i})$. Thus, $\Cbar_{j}$ can have no neighbors of $\Cbar_{i}$ as its elements. Hence, $\Cbar_{i}$ and $\Cbar_{j}$ are independent.
\end{proof}
 
To prove that \cref{alg:ics} outputs an ICS, we need to show that $x$ returned at the end of the algorithm satisfies LCP$_\delta(G)$ conditions \cref{eqn:lcp1,eqn:lcp2,eqn:lcp3} for all $i \in V$ and that $\sigma(x)$ is a union of independent cliques. The fact that $\sigma(x)$ is a union of independent cliques is evident from \cref{lem:complete}. To show that the LCP$_\delta(G)$ conditions  \cref{eqn:lcp1,eqn:lcp2,eqn:lcp3} are satisfied, we show for the $x$ generated when we exit from the algorithm, these conditions are satisfied for $\delta \in [\gamma(G),1)$. We do this in  \cref{thm:cliqexist}.

Before we discuss \cref{thm:cliqexist}, we have a lemma describing the conditions on the neighborhood of a vertex $j$ so that with such a neighborhood and corresponding discounted sum of closed neighborhood, $x_j = 0$ for the LCP$_\delta(G)$ conditions \cref{eqn:lcp1,eqn:lcp2,eqn:lcp3} to be satisfied for $j$.

\begin{lemma}\label{lem:cont0}[Two Clique Lemma]
Consider a vertex $i$ and cliques $C_1=K_n$ and $C_2=K_m$ of graph $G$ such that $C_1 \subset N_G(i)$ and there exists another clique $C_2 = K_m$ for some $m$ with $C_1 \cap C_2 = \emptyset$ such that $V(C_2) \cap {N_G(i)} \neq \emptyset$. Let $x \in \Real^n$ such that 
\[
x_{j} =
  \begin{cases}
   \frac{1}{1+(n-1)\delta}        & {\rm if }\ j \in C_1, \\
   \frac{1}{1+(m-1)\delta}         & {\rm if } \ j \in C_2, \\
   0         & {\rm if } \ j = i.
  \end{cases}
\]
and $x_j \geq 0 \ \forall j \in N_G(i)$. Then, if $\delta \in [\gamma(G),1)$, we have that for $i$ the LCP$_\delta(G)$ conditions \cref{eqn:lcp1,eqn:lcp2,eqn:lcp3} are satisfied by the above $x$, where
\begin{equation}\label{eqn:gamma_def}
\gamma(G) = \frac{\omega(G)-3 + \sqrt{(\omega(G)-3)^2 + 4(\omega(G)-1)}}{2(\omega(G)-1)},
\end{equation}
and $\omega(G)$ represents the size of the largest clique in the graph $G$.
\end{lemma}
\begin{proof}
See \cref{appen:lem:cont0}.
\end{proof}

\begin{remark} Note that for any graph with at least one edge, the value of $\omega(G)$ is at least $2$ and hence the value of $\gamma(G)$ is at least $\frac{\sqrt{5}-1}{2}$, which is  the golden ratio. For trees $G$, $\gamma(G)=\frac{\sqrt{5}-1}{2}.$ We do not know of any deeper significance or interpretations of the appearance of the golden ratio in this problem.
\end{remark}

We say that vertex $i$ is {\it fully connected} to a clique $C$ if all the vertices of $C$ are neighbors of $i$, i.e. $j \in N_G(i) \ \forall j \in C$.

\begin{theorem}\label{thm:cliqexist}
For \textup{LCP}$_\delta(G)$ with $\delta \in [\gamma(G),1)$ where 
$\gamma(G)$ is defined as in \cref{eqn:gamma_def},
\cref{alg:ics} returns an ICS whose support is a union of $\alpha(G)$ independent cliques.
\end{theorem}
\begin{proof}
Recall from \cref{alg:ics} that $S = \{1,2,\hdots,{|S|}\}$ is a maximum indepdendent set of $G$, $L = \{l \in V| \ l {\rm \ has \ only \ one \ neighbor \ in} \ S\}$. Observe that $C_{i} = N_{\Gbar}(i) \cap L$, $\Cbar_{i} = C_{i} \cup \{i\}$ and $Q_{i} = N_{\Gbar}(\Cbar_{i})$ get defined with respect to $\Gbar$ at the start of iteration $i$ in \cref{line:C_def}. We will show that after $|S|$ iterations, all the vertices satisfy the LCP$_\delta(G)$ conditions \cref{eqn:lcp1,eqn:lcp2,eqn:lcp3} for $\delta \in [\gamma(G),1)$. 

We first make a few preliminary observations.
\begin{enumerate}
\item[(a)]  The set of vertices removed  in \cref{alg:ics} $R = \bigcup_{i  = 1}^{|S|} (\Cbar_{i} \cup Q_{i}) $ equals $V$ by \cref{rem:algremV}. Thus $V$ can be divided into two disjoint sets, $V=R_C \cup R_Q$ where
$R_C \coloneqq \bigcup_{i  = 1}^{|S|} \Cbar_{i}$ and $R_Q \coloneqq \bigcup_{i  = 1}^{|S|} Q_{i}$. Moreover, $Q_i$'s are disjoint because if vertex $j$ is chosen after vertex $i$ in the for loop, $Q_{j} \subseteq V \backslash (Q_{i} \cup\Cbar_{i})$. Similarly $\Cbar_i$'s are disjoint.
\item[(b)]  We have that $R_Q \cap S = \emptyset$. This can be seen as follows. For some $i \in S$, let $q \in Q_{i}$. Now, $Q_{i} \subset N_G(i) \cup N_G(C_{i})$. If $q \in N_G(i)$, it cannot be an element of $S$ since $S$ is an independent set. Suppose $q \in N_G(c)$ for some $c \in C_{i}$. Now  recall that $N_G(c) \cap S = \{i\}$ by definition of $C_{i}$, whereby $q \not\in S$ unless $q = i$, which is not true by definition of $Q_i$. Hence, $R_Q \cap S = \emptyset$.
\item[(c)] \cref{alg:ics} returns $x_{j}  = 0 \ \forall j \in R_Q$. To prove this note that each $x_j$ for $j \in R_Q$ is initialized as 0 since $R_Q \cap S=\emptyset$. Furthermore, in each iteration, we assign non-zero values only to vertices in $R_C$.  
\item[(d)] Let $c \in \Cbar_i$ for $1\leq i\leq |S|$. \cref{alg:ics} returns $x_j = 0 \ \forall j \in N_{G}(c) \backslash N_{\Cbar_{i}}(c)$. To see this note that $N_{G}(c)$ $\backslash N_{\Cbar_{i}}(c)\subset N_{G}(\Cbar_{i})$ and since all cliques are independent, we have $N_{G}(\Cbar_{i}) \subset R_Q$. The claim then follows from part (c).
\item[(e)] Consider $q \in Q_{i}$ for some $1\leq i \leq |S|$. For the vertex $q$, define its set of ``protective vertices''  as $P_{q} \coloneqq (N_{G}(q) \cap S) \backslash \{i\}$. Then $P_{q}$ is always non-empty. The definition of $P_q$ is valid because $Q_j$'s are disjoint. 
To show $P_q\neq \emptyset$ consider two cases: if $q \in N_G(i)$ then since $q\notin \Cbar_{i}$,  $q$ must have at least two neighbors in $S$, one of which is $i$. On the other hand, if $q \in Q_{i} \backslash N_G(i) $, it needs to have at least one neighbor in $S$ for $S$ to be a maximum independent set and since $q$ is not a neighbor of $i$, 
%since no such neighbor exists in $\Cbar_{i} \backslash \{i\}$, 
there must be a neighbor of $q$ in $S\backslash \{i\}$. 
\end{enumerate}

Now we claim that, with the returned value of $x$ conditions  \cref{eqn:lcp1,eqn:lcp2,eqn:lcp3} are satisfied for all the vertices in $V$. First we show this for vertices in $R_C$. 
For any vertex $c \in \Cbar_{i}$, $1\leq i\leq |S|$, for the $x$ returned by \cref{alg:ics}, we have 
\begin{align}
 \C_c(x) &= x_{c} + \delta \sum_{j \in N_{G}(c)}x_j \label{eqn:cliqexistcomp1}\\
 &= x_{c} + \delta \sum_{j \in N_{\Cbar_{i}}(c)}x_j + \delta \sum_{j \in N_{G}(c) \backslash N_{\Cbar_{i}}(c)}x_j\label{eqn:cliqexistcomp2}\\
& = \frac{1}{1+(|\Cbar_{i}|-1)\delta} + \frac{\delta(|\Cbar_{i}|-1)}{1+(|\Cbar_{i}|-1)\delta} + 0 = 1 \label{eqn:cliqexistcomp}
\end{align}
\cref{eqn:cliqexistcomp} is justified because of (d) above.
 Thus, for any vertex $c \in \Cbar_{i}$, $\C_c(x) = 1$ and the LCP conditions \cref{eqn:lcp1,eqn:lcp2,eqn:lcp3} are satisfied. It follows that \cref{eqn:lcp1,eqn:lcp2,eqn:lcp3} holds for all vertices in $R_C$.

We now show that \cref{eqn:lcp1,eqn:lcp2,eqn:lcp3} hold for all vertices in $R_Q$. Let $q \in Q_i$ for $1\leq i \leq |S|$ and let $P_{q} = \{p_1,p_2,\cdots,p_r\}$ be its set of protective vertices. $P_q \neq \emptyset$ by (e) above. We now have the following cases:\\

\noindent \textit{Case 1:} $\Cbar_{i} \subset N_{G}(q)$. \\
Recall that $i$ is not in the set of protective vertices. Since $\Cbar_{i} \subset N_{G}(q)$ we have that $q$ is fully connected to $\Cbar_i$ and adjacent to at least one protective vertex, say $p_s \in P_q$. \cref{alg:ics} assigns $x_{p_s} = \frac{1}{1+(|\Cbar_{p_s}|-1)\delta}$ and $x_t=\frac{1}{1+(|\Cbar_i|-1)\delta}$ for all $t \in \Cbar_i.$ Thus using the Two Clique Lemma (\cref{lem:cont0}) we get that for $\delta \in [\gamma(G),1)$, $x$ satisfies the LCP conditions \cref{eqn:lcp1,eqn:lcp2,eqn:lcp3} for vertex $q$. \\
%\cred{Note that in the cases when $\Cbar_t = \{t\}$, $x_t = \frac{1}{1+(|\Cbar_{t}|-1)\delta} = 1$ and such cases are included in the above argument. Why is this line needed?}

\noindent \textit{Case 2:} $\Cbar_{p_s} \subset N_{G}(q)$ for some $s \in  \{1,2,\hdots,r\}$. \\
Since $\Cbar_{p_s} \subset N_{G}(q)$ we have that $q$ is fully connected to $\Cbar_{p_s}$ and adjacent to some vertex $j \in \Cbar_i$. \cref{alg:ics} assigns $x_t = \frac{1}{1+(|\Cbar_{p_s}|-1)\delta} \ \forall t \in \Cbar_{p_s}$ and  $x_{j} =  \frac{1}{1+(|\Cbar_{i}|-1)\delta}$. Thus,  using the Two Clique Lemma (\cref{lem:cont0}) for $\delta \in [\gamma(G),1)$, $x$  satisfies the LCP conditions \cref{eqn:lcp1,eqn:lcp2,eqn:lcp3} for $q$. \\
%Note that in the cases when $\Cbar_t = \{t\}$, $x_t = \frac{1}{1+(|\Cbar_{t}|-1)\delta} = 1$ and such cases are included in the above argument. 

\noindent \textit{Case 3:} $\Cbar_{p_s} \backslash N_{G}(q) \neq \emptyset \ \forall s \in  \{1,2,\cdots,r\}$ and $\Cbar_{i} \backslash N_{G}(q) \neq \emptyset$.\\
We will show that this case is not possible. Define $D_{s} = \Cbar_{p_s} \backslash N_{G}(q) \ \forall s \in  \{1,2,\cdots,$ $r\}$. Note that all $D_{s}, s=1,\hdots,r$ are independent cliques  since $D_{s} \subset \Cbar_{p_s}$ and each $\Cbar_{p_s}$ are independent cliques from \cref{lem:complete}.

Now, choose a maximal independent set $D  = \{d_1,d_2,\cdots,d_r\}$ of $G_{\cup_{s = 1}^{r}D_{s}}$. Each vertex in $D$ is from a different clique $D_{s}$ and there is one vertex from each clique; thus, $|D| = r$. Also, let $l \in \Cbar_{i} \backslash N_{G}(q)$. We define 
\[
S'' \coloneqq \Stilde \cup D \cup \{q,l\}, \quad \Stilde := S \backslash (P_{q}\cup \{i\}),
\]
and show that $S''$ is an independent set with  $|S''| > |S|$. To see that $S''$ is an independent set, we check the independence of each of the pairs of sets in the union above. 
$\Stilde$ and $D$ are indepdendent because $D \subset L \cap N_G(P_{q})$ and hence $N_G(D) \cap S \subset P_{q}$ whereby $N_G(D)\cap \Stilde = \emptyset$. $\Stilde$ and $q$ are indepdendent because $N_G(q) \cap S \subset P_{q}\cup \{i\}$. Finally, $N_G(l) \cap S = \{i\}$ by definition whereby $\Stilde$ and $l$ are independent. Also, $D \cap N_{G}(q) = \emptyset$ by definition and $D \cap N_{G}(l) \subset = \emptyset$ since $\Cbar_t, t=1,\hdots,|S|$ are independent. Finally, $N_G(q) \cap \{l\} = \emptyset$ by definition. Thus, $S''$ is an independent set. Moreover, $|S''| > |S|$, which gives a contradiction since $S$ was assumed to be a maximum independent set. Hence, \textit{Case 3} is not possible.

Thus, since $q$ was chosen to be an arbitrary vertex in $R_Q$, the LCP conditions \cref{eqn:lcp1,eqn:lcp2,eqn:lcp3} are satisfied for all vertices in $R_Q$. Thus, we have exhausted all cases and constructively shown the existence of an independent clique solution in any graph.
\end{proof}

\begin{example}
In this example, we discuss the tightness of the bound $\delta \geq \gamma(G)$ in \cref{thm:cliqexist}. We will show that for a path of length $4$ as shown in \cref{fig:tight_gamma}, the vector returned by \cref{alg:ics} is not a solution when $\delta < \gamma(G) = \frac{\sqrt{5}-1}{2}$. \cref{alg:ics} returns $x =\Big(1,0,\frac{1}{1+\delta},\frac{1}{1+\delta}\Big)$. Now,  $\sigma(x)$ is a union of independent cliques whose vertex sets are given by $\{1\}$ and $\{3,4\}$ respectively. For $x$ to be an ICS, $\C_2(x) \geq 1$. Now, $\C_2(x) = \delta \Big(1 + \frac{1}{1+\delta} \Big) < 1$ for $0 < \delta < \frac{\sqrt{5}-1}{2}$. This example shows that there are cases where $\delta \in [\gamma(G),1)$ is not only sufficient but also necessary. 
\end{example}

%\begin{remark}
%For any tree $T$, $\omega(T) = 2$ and hence $\gamma(T) = \frac{\sqrt{5}-1}{2}$, which is the golden ratio. Since $\gamma(G)$ is an increasing function in $\omega(G)$, the golden ratio is an absolute minimum value of $\gamma(G)$ for non-degenerate graphs. 
%\end{remark}

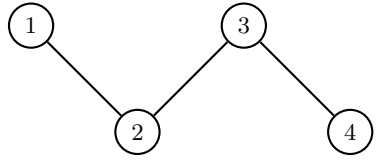
\begin{figure}
\centering
\begin{tikzpicture}[auto, node distance=2cm, every loop/.style={},
                    thick,main node/.style={circle,draw,font=\sffamily\small\bfseries}]

  \node[main node] (1) {$1$};
  \node[main node] (2) [below right of=1] {$2$};
  \node[main node] (3) [above right of=2] {$3$};
  \node[main node] (4) [below right of=3] {$4$};

  \path[every node/.style={font=\sffamily\small}]
    (1) edge node [left] {} (2)
    (2) edge node [right] {} (3)
    (3) edge node [right] {} (4);
\end{tikzpicture}
\caption{An example graph showing tightness of $\gamma(G)$.}
\label{fig:tight_gamma}
\end{figure}

\cref{thm:cliqexist} proves that for a given $\delta \in [\gamma(G),1)$, an ICS (say $x_\delta$) exists. 
We now show that there exists an ICS of $\LCP_{\delta'}(G)$ where $\delta' \geq \delta$, such that its support is same as that of $x_\delta$. To see this, we  note that given a set of independent cliques $\mathcal{K} = \{C_1,\dots,\C_n\}$, the only possible ICS for some $\delta \in [\gamma(G),1)$ with its support as $\mathcal{K}$ is $x^\mathcal{K}(\delta)$ given by
\begin{equation}\label{eqn:candidateICS}
 x^\mathcal{K}_i(\delta) = \begin{cases}
                    \frac{1}{1 + (|C_j|-1)\delta}, \ & \ {\rm if} \ i \in C_j, \ C_j \in \mathcal{K}\\
                    0, \ & \ {\rm otherwise}.
                   \end{cases}
\end{equation}
\begin{comment}
This is because all the neighbors of each clique $C_i \in \mathcal{K}$ are such that $x^\mathcal{K}_j  = 0 \ \forall j \in N_G(C_i) \ {\rm for} \ C_i \in \mathcal{K}$ for all possible candidate solutions $x^\mathcal{K}$. Hence, for $c \in C_i$ for some $C_i \in \mathcal{C}$, 
\begin{align*}
 \C_c(x^\mathcal{C}) &= x_{c} + \delta \sum_{j \in N_{G}(c)}x_j \\
 &= x_{c} + \delta \sum_{j \in N_{C_i}(c)}x_j + \delta \sum_{j \in N_{G}(c) \backslash N_{C_i}(c)}\underbrace{x_j}_{= 0}\\
 &=  x_{c} + \delta \sum_{j \in N_{C_i}(c)}x_j \\
\C_c(x^\mathcal{C}) &= \C_c((x^\mathcal{C})_{C_i}),
\end{align*}
where $(x^\mathcal{C})_{C_i}$ is the vector $x^\mathcal{C}$ restricted to index set $C_i$. Hence, we can consider $(x^\mathcal{C})_{C_i}$ as the full support solution in LCP$_\delta(C_i)$. It can be easily seen that for any clique $C_i$, the only full support solution is given by 
\[
 y_j = \frac{1}{1+(|C_i|-1)\delta} \ \forall j \in C_i.
\]
Thus, $x^\mathcal{C}$ given in \cref{eqn:candidateICS} is the only possible ICS with its support as $\mathcal{K}$.
\end{comment}
This gives us the following result.
\begin{corollary}\label{cor:valid}
 For a given set of independent cliques $\mathcal{K} = \{C_1,\dots,C_n\}$, if $x^{\Kscr}$ given by \cref{eqn:candidateICS} is an ICS for some $\delta > 0$, then there exists an ICS with the same support $\mathcal{K}$ for any $\delta' \in [\delta,1)$.
\end{corollary}

\begin{proof}
 Consider $x^\mathcal{K}(\cdot):[\delta,1)\rightarrow [0,1]^{|V|}$ such that,
 \begin{equation}\label{eqn:icsfundef}
   x^\mathcal{K}_i(\delta') = \begin{cases}
                    \frac{1}{1 + (|C_j|-1)\delta'}, \ & \ {\rm if} \ i \in C_j, \ C_j \in \mathcal{K}\\
                    0, \ & \ {\rm otherwise},
                   \end{cases}  
 \end{equation}
 for $\delta' \in [\delta,1)$. Then, we claim that for any given $\delta' \in [\delta,1)$, $x^\mathcal{K}(\delta')\in {\rm SOL}_{\delta'}(G)$. 
 
We know that for $\delta' = \delta$, $x^\mathcal{K}(\delta)$ is same as the one defined in \cref{eqn:candidateICS} and $x^\mathcal{K}(\delta) \in {\rm SOL}_{\delta}(G)$. Thus, for $c \in C_i$ for some $C_i \in \mathcal{K}$, $\C_c(x^\mathcal{K}(\delta)) = 1$ following the analysis in \cref{eqn:cliqexistcomp1,eqn:cliqexistcomp2,eqn:cliqexistcomp}. Since, the analysis in \cref{eqn:cliqexistcomp1,eqn:cliqexistcomp2,eqn:cliqexistcomp} is independent of the value of $\delta$, $\C_c(x^\mathcal{K}(\delta')) = 1$ for $c \in C_i$ for some $C_i \in \mathcal{K}$ for any $\delta' \in [\delta,1)$.
 
 For any vertex $j \in V \backslash \mathcal{K}$, let $|N_G(k) \cap C_i| = n_i$ $\forall i$ such that $C_i \in \mathcal{K}$. Then, 
 \[
  \C_j(x^\mathcal{K}(\delta')) = \delta'\sum_{C_i \in \mathcal{K}}^{} \frac{n_i}{1 + (|C_i|-1)\delta'}.
 \]
We know that $\C_j(x^\mathcal{K}(\delta)) \geq 1$ since $x^\mathcal{K}(\delta) \in {\rm SOL}_{\delta}(G)$. Since,  $\C_j(x^\mathcal{K}(\delta'))$ is an increasing function in $\delta'$, $\C_j(x^\mathcal{K}(\delta')) \geq 1, \ \forall \delta' > \delta$. Thus, $x^\mathcal{K}(\delta') \in {\rm SOL}_{\delta'}(G)$ for $\delta' \in [\delta,1)$.
\end{proof}

\begin{remark}\label{rem:algo_funsol}
 Consider a maximum independent set $S = \{1,\dots,{|S|}\}$ and a union of independent cliques $\mathcal{K} = \{\Cbar_{1},\dots,\Cbar_{{|S|}}\}$ obtained from \cref{alg:ics} and $x^\mathcal{K}(\cdot):[\gamma(G),1)\rightarrow [0,1]^{|V|}$ as defined in \cref{eqn:icsfundef}. Then, $x^\mathcal{K}(\delta)\in {\rm SOL}_{\delta}(G)$ for $\delta \in [\gamma(G),1)$. This follows since \cref{thm:cliqexist} shows that the output of \cref{alg:ics} viewed as a function of $\delta$ is a valid ICS of $\LCP_\delta(G)$ for any $\delta \in [\gamma(G),1)$.
\end{remark}

%\cred{Next, we claim that for any set of independent cliques $\mathcal{K}$ such that $x^\mathcal{K}(\delta)$ as defined in \cref{eqn:icsfundef} is in ${\rm SOL}_{\delta}(G)$, we can find an independent set $S$ such that $\lim_{\delta' \rightarrow 1}||x^\mathcal{K}(\delta')||_1 = |S|$. Note that the limit is valid because by \cref{cor:valid}, $x^\mathcal{K}(\delta')$ is well defined for $\delta' \in [\delta,1)$ and $x^\mathcal{K}(\delta') \in {\rm SOL}_{\delta'}(G)$ for $\delta' \in [\delta,1)$. We call such independent sets \textit{norm preserving} with respect to $\mathcal{K}$. This is found by choosing one vertex in each clique and since each clique is independent, the set so formed is also independent. 
%%Since the independent sets are found by choosing one vertex from each clique, 
%It follows that the size of any norm preserving independent set is equal to the number of independent cliques in $\mathcal{K}$. -- do we need this concept of norm preserving ind sets?}
%Note that for a set of independent cliques constructed using \cref{alg:ics}, one such independent set is the one chosen in \cref{line:init1} of the algorithm. 

We now show that there exists a set of independent cliques $\mathcal{K}$ with $|\Kscr|=\alpha(G)$, such that $x^\mathcal{K}(\delta)$ is a solution of maxICS$_\delta(G)$ for $\delta \in [\eta(G),1)$. Here, $\eta(G)$ is defined as,
\begin{equation}\label{eqn:etadef}
 \eta(G) =  \max\Big{\{}\frac{\omega(G)-3 + \sqrt{(\omega(G)-3)^2 + 4(\omega(G)-1)}}{2(\omega(G)-1)},\frac{\alpha(G)(\omega(G)-1) - \omega(G)}{\alpha(G)(\omega(G)-1)}\Big{\}}.
\end{equation}

\begin{proposition}\label{prop:cliqmaxind}
 For $\delta \in [\eta(G),1)$, there exists a union of independent cliques $\Kscr$ and a solution   $x^\mathcal{K}(\delta)$ of {\rm maxICS}$_\delta(G)$  such that $|\mathcal{K}| = \alpha(G)$.
\end{proposition}
\begin{proof}
We show that for any two sets of independent cliques $\mathcal{K}_1 = \{C^1_1,\dots,C^1_{|\mathcal{K}_1|}\}$ and $\mathcal{K}_2 = \{C^2_1,\dots,C^2_{|\mathcal{K}_2|}\}$ such that $|\mathcal{K}_1| = \alpha(G)$ and $|\mathcal{K}_1| > |\mathcal{K}_2|$, we have $||x^{\mathcal{K}_1}(\delta)||_1 \\ \geq ||x^{\mathcal{K}_2}(\delta)||_1$ for all $\delta \in [\eta(G),1)$. Note that such a $\mathcal{K}_1$ exists, since $\eta(G) \geq \gamma(G)$, by using \cref{rem:algo_funsol}. Also, $||x^{\mathcal{K}_t}(\delta)||_1 = \sum_{i = 1}^{|\mathcal{K}_t|}\frac{|C^t_i|}{1 + (|C^t_i|-1)\delta}$ for $t = 1,2$. Each term in the sum is increasing with $|C^t_i|$. Thus, the least value of $||x^{\mathcal{K}_1}(\delta)||_1$ is when $|C^1_i| = 1, \forall C^1_i \in \mathcal{K}_1$, which is $\alpha(G)$. Similarly, the maximum value of $||x^{\mathcal{K}_2}(\delta)||_1$ is when $|C^2_i| = \omega(G), \forall C^2_i \in \mathcal{K}_2$ and $|\mathcal{K}_2| = \alpha(G) - 1$, which is $\frac{\omega(G)(\alpha(G)-1)}{1 + (\omega(G)-1)\delta}$. It is easy to see that 

$$\delta \geq \kappa(G) \coloneqq \frac{\alpha(G)(\omega(G)-1) - \omega(G)}{\alpha(G)(\omega(G)-1)} \iff \alpha(G) \geq \frac{\omega(G)(\alpha(G)-1)}{1 + (\omega(G)-1)\delta}.$$

Thus, $|x^\mathcal{K}_1(\delta)| \geq |x^\mathcal{K}_2(\delta)|$ for any $\mathcal{K}_2$ with $|\mathcal{K}_2| < \alpha(G)$. Hence, there must exist a solution of maxICS$_\delta(G)$ with its support as a union of $\alpha(G)$ independent cliques.
\end{proof}

\begin{example}
In this example, we discuss the tightness of the bound $\delta \geq \eta(G)$. Note that $\eta(G)$ is defined as the maximum amongst two terms, the first of which is $\gamma(G)$. This term comes in because we assume the existence of a set of indepdendent cliques $\mathcal{K}_1$ such that $|\mathcal{K}_1| = \alpha(G)$ for which we refer to \cref{thm:cliqexist} which requires $\delta \in [\gamma(G),1)$ and is tight for certain graphs. The second term in the maximum function in the definition of $\eta(G)$ is $\kappa(G)$. We now show the tightness of this term with the graph in \cref{fig:tight_eta} as an example. Note that $\omega(G) = 2$ and $\alpha(G) = 6$, hence $\gamma(G) = \frac{\sqrt{5} - 1}{2}$ and $\kappa(G) = \frac{2}{3}$ giving $\eta(G) = \kappa(G)$. The maximum independent set is $\{1,3,5,7,9,11\}$. Also, $\mathcal{K}_1 = \{1,3,5,7,9,11\}$ and consider $\mathcal{K}_2 = \{\{1,2\},\{4,5\},\{6,7\},\{8,9\},\{10,11\}\}$. $||x^{\mathcal{K}_1}_\delta(G)||_1 = 6$ and the solution is valid for $\delta \geq \frac{1}{2}$. Also, $||x^{\mathcal{K}_2}_\delta(G)||_1 = \frac{10}{1+\delta}$ and it is valid for $\delta \geq \frac{1}{4}$. Now, for $\delta < \frac{2}{3} = \frac{\alpha(G)(\omega(G)-1) - \omega(G)}{\alpha(G)(\omega(G)-1)}$, we have $\frac{10}{1+\delta} > 6$ and hence $||x^{\mathcal{K}_2}_\delta(G)||_1 > ||x^{\mathcal{K}_1}_\delta(G)||_1$. But we want $||x^{\mathcal{K}_2}_\delta(G)||_1 \leq ||x^{\mathcal{K}_1}_\delta(G)||_1$ in \cref{prop:cliqmaxind} for which it is not only sufficient but also necessary to have $\delta \geq \frac{\alpha(G)(\omega(G)-1) - \omega(G)}{\alpha(G)(\omega(G)-1)}$ in this example. Thus, since the both the terms over which the maximization takes place in the definition of $\eta(G)$ is tight, it is a tight bound for \cref{prop:cliqmaxind}.

\end{example}

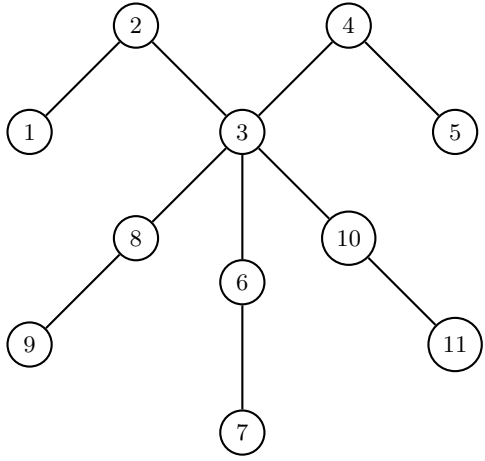
\begin{figure}
 \centering
 \begin{tikzpicture}[auto, node distance=2cm, every loop/.style={},
                    thick,main node/.style={circle,draw,font=\sffamily\small\bfseries}]

  \node[main node] (1) {$1$};
  \node[main node] (2) [above right of=1] {$2$};
  \node[main node] (3) [below right of=2] {$3$};
  \node[main node] (4) [above right of=3] {$4$};
  \node[main node] (5) [below right of=4] {$5$};
  \node[main node] (6) [below of=3] {$6$};
  \node[main node] (7) [below of=6] {$7$};
  \node[main node] (8) [below left of=3] {$8$};
  \node[main node] (9) [below left of=8] {$9$};
  \node[main node] (10) [below right of=3] {$10$};
  \node[main node] (11) [below right of=10] {$11$};

  \path[every node/.style={font=\sffamily\small}]
    (1) edge node [right] {} (2)
    (2) edge node [right] {} (3)
    (3) edge node [right] {} (4)
    (4) edge node [right] {} (5)
    (3) edge node [right] {} (6)
    (3) edge node [right] {} (8)
    (3) edge node [right] {} (10)
    (6) edge node [right] {} (7)
    (9) edge node [right] {} (8)
    (11) edge node [right] {} (10);
\end{tikzpicture}
\caption{An example graph showing tightness of $\eta(G)$.}
\label{fig:tight_eta}
\end{figure}

%%%%%%%%%%%%%%%%%%%%%%%%%%%%%%%%%%%%%%%%%%%%%%%%%%%%%%%%%%%%%%%%%%%%%%%%%%%%%%%%%%%%%%%%%%%%%%%%%%%%%%%%%%%%%%%%%%%%%%%%%%%%%%%%%%%%%%%%%%%%%%%%%%%%%%%%%%
\subsection{Monotonicity of the Maximum $\ell_1$ Norm amongst ICSs}\label{subsec:monotonic}

In this section, we show that the solutions of {\rm maxICS}$_\delta(G)$ for some $\delta \in [\eta(G),1)$ are monotonic with respect to the induced subgraph relation. We first prove a lemma which gives conditions under which a graph $G'$ formed by adding a vertex to a graph $G$ admits an ICS $(x^\mathcal{K}(\delta),0)$ where $x^\mathcal{K}(\delta)$ is an ICS of $G$ and $\mathcal{K}$ are a set of independent cliques in $G$. In other words, an ICS of $G'$ can be formed by appending a 0 corresponding to the new vertex to an ICS of $G$.

\begin{lemma}\label{lem:staysol}
Let $x^\mathcal{K}(\delta) \in \textup{ICS}_\delta(G)$ where $\mathcal{K}$ are a set of independent cliques in $G$. Let $ \delta \in [\gamma(G'),1)$, where $G'$ is such that 
\begin{enumerate}
 \item $V(G') = V(G) \cup \{i\}$,
 \item $\exists C \in \mathcal{K}$ s.t. $(i,j) \in E(G') \ \forall j \in C$,
 \item $\exists C' \in \mathcal{K}, C' \neq C$ s.t. $(i,j) \in E(G')$ for some  $j \in C'$.
\end{enumerate}
Then $x^\mathcal{K}(\delta)$ defined as
\[
x^\mathcal{K}_j(\delta) =
  \begin{cases}
   x^\mathcal{K}_j(\delta)        & {\rm if }\ j \in V(G), \\
   0         & {\rm if } \ j \ = i.
  \end{cases}
\]
is an ICS of \textup{LCP}$_\delta(G')$.
\end{lemma}
\begin{proof}
In $G'$, $i$ is a vertex such that $C \subset N_{G'}(i)$ and $C' \cap N_{G'}(i) \neq \emptyset$. Thus, using Two Clique Lemma (\cref{lem:cont0}), $i$ satisfies the LCP$_\delta(G)$ conditions \cref{eqn:lcp1,eqn:lcp2,eqn:lcp3} for $\delta \in [\gamma(G'),1)$.
\end{proof}

We now show the main result of this subsection showing the monotonicity of maximum $\ell_1$ norm amongst Independent Clique Solutions. In the following lemma, we have a graph $G'$ which is constructed by adding a vertex $i$ to $G$ and it is shown that the maximum $\ell_1$ norm amongst Independent Clique Solutions of $G'$ is greater than that amongst $G$. Recall that \cref{prop:cliqmaxind} showed that there always exists a union of independent cliques $\mathcal{K}$ of $G$ with $|\mathcal{K}| = \alpha(G)$ such that $x^\mathcal{K}_{\max}(\delta)$ as defined in $\cref{eqn:icsfundef}$ is a solution of \textup{{\rm maxICS}}$_\delta(G)$.

\begin{lemma}\label{lem:cliqorder}
Consider two graphs $G(V,E)$ and $G'(V',E')$, where $V' = V \cup \{i\}$ and  $G'_V = G$. Let $x^\mathcal{K}(\delta)$ and $x^{\mathcal{K}'}(\delta)$ be solutions of {\rm maxICS}$_\delta(G)$ and {\rm maxICS}$_\delta(G')$ respectively, where $\mathcal{K}$ and $\mathcal{K}'$ are sets of independent cliques of $G$ and $G'$ such that $|\mathcal{K}| = \alpha(G)$ for some $\delta \in [\eta(G'),1)$ with $\eta(\cdot)$ is defined as in \cref{eqn:etadef}. Then, $||x^\mathcal{K}(\delta)||_1$ $\leq ||x^{\mathcal{K}'}(\delta)||_1$.
\end{lemma}

\begin{proof}
We divide this proof into three cases depending on the neighborhood of the vertex $i$ which is in $G'$ but not in $G$. Let $\mathcal{K} = \{C_1,C_2,\dots, C_{|\mathcal{K}|}\}$. The first case is that $i$ is fully connected to some clique $C_j \in \mathcal{K}$ and has at least one more edge connecting at least one other clique $C_k \in \mathcal{K}$ with $k \neq j$. In this case, we explicitly construct a solution of \textup{{\rm maxICS}}$_\delta(G')$ with $\ell_1$ norm at least as much as $x^\mathcal{K}(\delta)$. The second case is when $i$ is not fully connected to any clique $C_i \in \mathcal{K}$. In this case, we show that the size of maximum indepdendent set of $G'$ is bigger than that of $G$ and hence we can construct an ICS of $G'$ using \cref{alg:ics} which has $\ell_1$ norm greater than that of $x^\mathcal{K}(\delta)$. The last case is when $N_{G'}(i)\cap \sigma(x^\mathcal{K}(\delta))$ is equal to a clique $C_j \in \mathcal{K}$. In this case, we divide into two subcases and construct separate solutions with $\ell_1$ norm greater than that of $x^\mathcal{K}(\delta)$.\\

\noindent \textit{Case 1:}  $i$ is fully connected to some clique $C_j \in \mathcal{K}$ and has at least one more edge connecting at least one other clique $C_k \in \mathcal{K}$ with $k \neq j$.

Let $y^\mathcal{K}(\delta) \in \Real^{|V(G')|}$ be defined as
\[
 y^\mathcal{K}_t(\delta) = \begin{cases}
                            x^\mathcal{K}_{\max,t}(\delta) & \text{if } t \in V(G) \\
                            0 &   \text{if } t = i.
                           \end{cases}
\]
The LCP$_\delta(G')$ conditions \cref{eqn:lcp1,eqn:lcp2,eqn:lcp3} are satisfied  $y^\Kscr(\delta)$ for all $j \in V(G)$ since $y^\mathcal{K}_i(\delta) = 0$ and they were already satisfied under $x^\mathcal{K}(\delta)$ in $G$. For $i$, by \cref{lem:staysol}, since $\delta \geq \eta(G') \geq \gamma(G')$ and since $i$ is fully connected to $C_j$ and has at least one more edge connecting $C_k$, the LCP$_\delta(G')$ conditions \cref{eqn:lcp1,eqn:lcp2,eqn:lcp3} are also satisfied for $i$ by the Two Clique Lemma. Thus, in this case, $||x^\mathcal{K}(\delta)||_1 = ||y^\mathcal{K}(\delta)||_1 \leq ||x^{\mathcal{K}'}(\delta)||_1$.\\

\noindent \textit{Case 2:} $C_j \not\subset N_{G'}(i)\cap \sigma(x^\mathcal{K}(\delta)) \ \forall C_j \in \mathcal{K}$, i.e. $i$ is not fully connected to any clique in $\mathcal{K}$.\\
In this case, we will a maximum independent set of the graph $G'$ which has size strictly bigger than the maximum independent set of $G$ thereby arriving at a contradiction. Let $S$ be a maximum independent set of $G$ formed by choosing one node from each independent clique in $\Kscr$. In $G'$ we construct an independent set with size $|S|+1$ in the following way. For all the cliques which have no vertex as a neighbor of $i$, choose any vertex to be in the independent set. For the cliques which intersect the neighborhood of $i$, choose a vertex in the clique which is not a neighbor of $i$ in the independent set. This gives us an independent set of size at least $|S|$. Now, by construction $i$ has no neighbors in the independent set found, hence including $i$ still preserves the independence. Thus we have an independent set, say $\widetilde{S}$, of size $|S|+1$, and the maximum independent size can increase at most by 1. Hence, $\alpha(G')=|S|+1$. Next, we construct an ICS of $G'$ using a maximum independent set and show that its $\ell_1$ norm is greater than that of $x^\mathcal{K}(\delta)$.

Using \cref{alg:ics} we can find a set of independent cliques $\widetilde{\mathcal{K}}$ in $G'$ and a corresponding ICS $x^{\widetilde{\mathcal{K}}}(\delta)$ of $\LCP_\delta(G')$  such that $|\widetilde{\mathcal{K}}| = \alpha(G')  =|S|+1 > \alpha(G)$. We now find conditions on $\delta$ such that $||x^\mathcal{K}(\delta)||_1 \leq ||x^{\widetilde{\mathcal{K}}}(\delta)||_1$. We consider the least possible value of $||x^{\widetilde{\mathcal{K}}}(\delta)||_1$ and the maximum possible value of $||x^\mathcal{K}(\delta)||_1$. Note that we have, $||x^{\widetilde{\mathcal{K}}}(\delta)||_1 = \sum_{C \in \widetilde{\mathcal{K}}}\frac{|C|}{1 + (|C|-1)\delta}$. Each term in the sum is increasing with $|C|$. Thus, the least value of $||x^{\widetilde{\mathcal{K}}}(\delta)||_1$ is when $|C| = 1, \forall C \in \widetilde{\mathcal{K}}$.  In this case $\norm{x^{\widetilde{\Kscr}}(\delta)}_1=\alpha(G')$. Similarly, the maximum value of $||x^{\mathcal{K}}(\delta)||_1$ is when $|C_i| = \omega(G), \forall i=1,\hdots,|\mathcal{K}|$, which gives $\norm{x^\Kscr(\delta)}_1=\frac{\omega(G)(\alpha(G')-1)}{1 + (\omega(G)-1)\delta}$. It is easy to see that 
\[
\delta \geq 1 - \frac{\omega(G)}{\alpha(G')(\omega(G)-1)} \implies \norm{x^{\widetilde{\Kscr}}(\delta)}_1\geq  \alpha(G') \geq \frac{\omega(G)(\alpha(G')-1)}{1 + (\omega(G)-1)\delta} \geq \norm{x^\Kscr(\delta)}_1.
\]
Thus, since $\delta \in [\eta(G'),1)$, we have 
\[
||x^\mathcal{K}(\delta)||_1 \leq ||x^{\widetilde{\mathcal{K}}}(\delta)||_1 \leq ||x^{\mathcal{K}'}(\delta)||_1.
\]

\noindent \textit{Case 3:} $N_{G'}(i)\cap \sigma(x^\mathcal{K}(\delta)) = C_j$ for some $C_j \in \mathcal{K}$, i.e., $i$ is fully connected to exactly one clique $C_j \in \mathcal{K}$ and has no edges to any other cliques $C_t \in \mathcal{K}$ for $t \neq j$.\\
Define a set of indepdendent cliques of $G'$,  $\widehat{\mathcal{K}}$, by replacing $C_j$ in $\mathcal{K}$ by $C_j \cup \{i\}$, i.e., $\widehat{\mathcal{K}} = (\mathcal{K} \backslash C_j) \cup \{C_j \cup \{i\}\}$.  Let $M = N_{G'}(C_j \cup \{i\})$, which is the neighborhood of the clique $C_j \cup \{i\}$. Note that any $m \in M$ has neighbors in at least one more clique other than $C_j$. For if not, then $N_{G}(m) \cap \sigma(x^\mathcal{K}(\delta)) \subseteq C_j$ and $\C_m(x^\mathcal{K}(\delta)) \leq \frac{|C_j|\delta}{1 + (|C_j|-1)\delta} < 1$, which is a contradiction since $x^\mathcal{K}(\delta)$ is a solution to LCP$_\delta(G)$.

Define $F = \{m \in M | C_l \subset N_{G'}(m) \cap \sigma(x^{\widehat{\mathcal{K}}}(\delta)), \ {\rm for \ some  \ } C_l \in \widehat{\mathcal{K}}\}$, where $x^{\widehat{\mathcal{K}}}(\delta)$ is defined as in \cref{eqn:candidateICS} for the collection of cliques $\widehat{\Kscr}$. $F$ is the set of neighbors of $C_j \cup \{i\}$ such that they are fully connected to at least one clique $C_l \in \widehat{\mathcal{K}}$. We now have two cases depending on whether $F= M$.\\

%Vertices in $M$ comprise of neighbors of $C_j$ and neighbors of $i$. Since the only clique in $\Kscr$ that $i$ has a neighbor with is $C_j$, all vertices in $M$ that are neighbors of $i$ are not in any of the cliques in $\Kscr$. Similarly vertices in $M$ that are neighbors of $C_j$ are not cliques in $\Kscr$ since these cliques are independent. 
%\cred{Thus, we have $x^\mathcal{K}_{\max,m}(\delta) = 0 \ \forall m \in M$. -- where is this being used? If not, drop the entire argument.}
%
%
%\cblue{need to discuss if $F$ can be empty}
%\\

\noindent \textit{Case 3a:} $F = M$.\\
In this case, we show that $x^{\widehat{\mathcal{K}}}(\delta)$ is a solution of LCP$_\delta(G')$. To claim this, we divide the vertex set $V'$ as 
\[
V' = (V' \backslash (F \cup C_j \cup \{i\})) \cup (F) \cup (C_j \cup \{i\}).
\] 
Consider $v \in V' \backslash (F \cup (C_j \cup \{i\}))$. It follows from \cref{eqn:candidateICS} that for such a $v$, we have $x^{\widehat{\mathcal{K}}}_i(\delta) = x^{\mathcal{K}}_{i}(\delta)$ for all $i \in \Nbar_{G}(v)$. Since $ x^{\mathcal{K}}(\delta)$ is a solution to $\LCP_{\delta}(G)$, the  conditions \cref{eqn:lcp1}-\cref{eqn:lcp3} for $\LCP_\delta(G')$  are satisfied for all $v \in V' \backslash (F \cup (C_j \cup \{i\}))$ by $x^{\widehat{\Kscr}}(\delta)$. Now consider $v \in F$. We note that if $v$ is such that $C_l$ is any clique in $\widehat{\mathcal{K}}$ other than $C_j \cup \{i\}$, then $v$ is fully connected to $C_l$ and has at least one neighbor in $C_j \cup \{i\}$. Hence, using the Two Clique lemma, \cref{lem:staysol}, the $\LCP_\delta(G')$ conditions are satisfied with $x^{\widehat{\mathcal{K}}}(\delta)$ for $v$. Next, if $C_l = C_j \cup \{i\}$,  since $v\in F=M$, we conclude that $v$ has neighbors in at least one clique other than $C_j$ in $\widehat{\mathcal{K}}$. Again, using the Two Clique Lemma, \cref{lem:staysol}, LCP conditions are satisfied for $v$ by $x^{\widehat{\Kscr}}(\delta)$. Finally for any node in $C_j \cup \{i\}$, the LCP$_\delta(G')$ conditions are satisfied by definition of $x^{\widehat{\mathcal{K}}}(\delta)$. Thus, for any vertex in $V'$,  the LCP$_\delta(G')$ conditions are satisfied with $x^{\widehat{\mathcal{K}}}(\delta)$. 
 
Note that, $||x^{\widehat{\mathcal{K}}}(\delta)||_1 - ||x^\mathcal{K}(\delta)||_1 = \frac{n+1}{1+n\delta} - \frac{n}{1+(n-1)\delta} \geq 0 \  \forall \delta < 1$. Thus, $||x^\mathcal{K}(\delta)||_1 \leq ||x^{\widehat{\mathcal{K}}}(\delta)||_1 \leq ||x^{\mathcal{K}'}(\delta)||_1$.\\

\noindent \textit{Case 3b:} $F \neq M$.\\
We show that in this case we can construct an independent set in $G'$ with size strictly greater than $\alpha(G)$. Let $p \in M \backslash F$. We know that $p$ has neighbors in at least one  clique other than $C_j$, say $C_o$, and neither clique is fully connected to $p$ since $p \not\in F$. Let $S$ be a maximum independent set of $G$ formed by choosing one node from each independent clique in $\Kscr$. 
Suppose the nodes chosen from $C_o$ and $C_j$ are $a$ and $b$, respectively. Now, consider  $\widetilde{S} = (S \backslash \{a,b\}) \cup \{c,d,p\}$ where $c \in C_o \backslash N_{G'}(p)$ and $d \in (C_j \cup \{i\}) \backslash N_{G'}(p)$. It can be observed that $\widetilde{S}$ is an indepdendent set and $|\widetilde{S}| > |S|$. Thus we have found an independent set of $G'$ of size greater than that of the maximum independent set of $G$. Using the argument in the second paragraph of \textit{Case 2}, we have a set of indepdendent cliques in $G'$, say $\widetilde{\mathcal{K}}$ generated from $\widetilde{S}$ by \cref{alg:ics}, such that $||x^\mathcal{K}(\delta)||_1 \leq ||x^{\widetilde{\mathcal{K}}}(\delta)||_1 \leq ||x^{\mathcal{K}'}(\delta)||_1$.

Thus, having considered all possible cases for the neighborhood of the node $i$, we have shown that $||x^\mathcal{K}(\delta)||_1 \leq ||x^{\mathcal{K}'}(\delta)||_1$.
\end{proof}

\subsection{Proof of the Main Results}\label{subsec:proof}

We have shown the existence of ICSs in any graph and the monotonicity of the solution of {\rm maxICS}$_\delta(G)$ with respect to the number of vertices in the graph. Next, we show that for $\delta \in [\eta(G),1)$, any solution of {\rm maxICS}$_\delta(G)$ is in fact also a solution of {\rm maxSOL}$_\delta(G)$. To show this, we first show that the $\ell_1$ norm of any ICS is greater than that of any full support solution of the LCP$_\delta(G)$. Then, by induction, we prove our main result.

\begin{proposition}\label{prop:cliqfull}
Let $\delta>0$ and consider  a graph $G= (V,E)$. The $\ell_1$ norm of any ICS is greater than that of any solution of  $\LCP_\delta(G)$ with support $V$.
\end{proposition}
\begin{proof} Let 
$V=\{1,2,\dots, n\}$. Let $x^\mathcal{K}$ be an ICS of $\LCP_\delta(G)$ and let $x$ be a solution of $\LCP_\delta(G)$ with support $V$. Clearly, for all $i\in V$,
\[
C_i(x) = x_i + \delta \sum_{j = 1}^{n}a_{ij}x_j = 1.
\]
Summing over $i \in \sigma(x^\mathcal{K})$, we get
\begin{equation}\label{eqn:cliqfull1}
\sum_{i \in  \sigma(x^\mathcal{K})}x_i = |\sigma(x^\mathcal{K})| - \delta\sum_{i \in \sigma(x^\mathcal{K})} \sum_{j = 1}^{n}a_{ij}x_j.
\end{equation}
Now,
\begin{align} 
\nonumber \sum_{i \in V} x_i &= \sum_{i \in \sigma(x^\mathcal{K})} x_i + \sum_{i \in \sigma(x^\mathcal{K})^c} x_i, 
\end{align}
whereby from \cref{eqn:cliqfull1},
\begin{align}
\nonumber \sum_{i \in V} x_i &= \sum_{i \in \sigma(x^\mathcal{K})^c} x_i + |\sigma(x^\mathcal{K})| - \delta\sum_{i \in \sigma(x^\mathcal{K})} \sum_{j = 1}^{n}a_{ij}x_j \ \ \ \ \\ 
%\nonumber &= |\sigma(x^\mathcal{K})| - \delta \sum_{i \in \sigma(x^\mathcal{K})} \sum_{j \in \sigma(x^\mathcal{K})}a_{ij}x_j  +  \sum_{i \in \sigma(x^\mathcal{K})^c} x_i - \delta\sum_{i \in \sigma(x^\mathcal{K})} \sum_{j \in \sigma(x^\mathcal{K})^c}a_{ij}x_j \
\label{eqn:cliqfull5}
&= |\sigma(x^\mathcal{K})| - \delta\sum_{i \in \sigma(x^\mathcal{K})} \sum_{j \in \sigma(x^\mathcal{K})}a_{ij}x_j  - \sum_{j \in \sigma(x^\mathcal{K})^c}x_j(\delta \sum_{i \in \sigma(x^\mathcal{K})}a_{ij} - 1).
\end{align}
Let $\mathcal{K} = \{C_j|1\leq j\leq m\}$ where $C_j$'s are disjoint cliques and let $n_j = |C_j|$. Then,
\begin{align*}
\sum_{i \in \sigma(x^\mathcal{K})} x_i^\mathcal{K} & = \sum_{i = 1}^{m}\frac{n_i}{1+(n_i-1)\delta}, \\ 
\text{and } |\sigma(x^\mathcal{K})| & = \sum_{i = 1}^{m}n_i.
\end{align*}
%Thus,
%\begin{equation*}
%|\sigma(x^\mathcal{K})| = \sum_{i \in \sigma(x^\mathcal{K})} x_i^\mathcal{K} + \sum_{i = 1}^{m}n_i - \sum_{i = 1}^{m}\frac{n_i}{1+(n_i-1)\delta}
%\end{equation*}
Also,
\begin{equation}\label{eqn:cliqfull9}
\delta\sum_{i \in \sigma(x^\mathcal{K})} \sum_{j \in \sigma(x^\mathcal{K})}a_{ij}x_j = \delta\sum_{i = 1}^{m}\frac{n_i(n_i-1)}{1+(n_i-1)\delta},
\end{equation}
which follows from the fact that amongst the vertices in $\sigma(x^\mathcal{K})$ only the ones which are in the complete networks are neighbors with each other. This gives,
\begin{align*}
|\sigma(x^\mathcal{K})| - \delta\sum_{i \in \sigma(x^\mathcal{K})} \sum_{j \in \sigma(x^\mathcal{K})}a_{ij}x_j &=  
%\sum_{i \in \sigma(x^\mathcal{K})} x_i^\mathcal{K} + 
\sum_{i = 1}^{m}n_i -% \sum_{i = 1}^{k}\frac{n_i}{1+(n_i-1)\delta} - \\
%& \ \  \ \  
\delta\sum_{i = 1}^{m}\frac{n_i(n_i-1)}{1+(n_i-1)\delta} \\ 
%& =  \sum_{i \in \sigma(x^\mathcal{K})} x_i^\mathcal{K} + \sum_{i = 1}^{k}n_i - \sum_{i = 1}^{k}\frac{n_i + \delta n_i(n_i-1)}{1+(n_i-1)\delta} \\ 
& =  \sum_{i \in \sigma(x^\mathcal{K})} x_i^\mathcal{K}.
\end{align*}
Thus, \cref{eqn:cliqfull5} becomes
\[
\sum_{i \in V} x_i +  \sum_{j \in \sigma(x^\mathcal{K})^c}x_j(\delta \sum_{i \in \sigma(x^\mathcal{K})}a_{ij} - 1)= \sum_{i \in \sigma(x^\mathcal{K})} x_i^\mathcal{K}.
\]

Now, consider $(\delta \sum_{i \in \mathcal{K}}a_{ij} - 1)$. Note that each vertex not in $\sigma(x^\mathcal{K})$ has at least $\lceil \frac{1}{\delta} \rceil$  neighbors in $\sigma(x^\mathcal{K})$, since it is a necessary condition for $x^\Kscr$ to be a solution. Thus, $(\delta \sum_{i \in \sigma(x^\mathcal{K})}a_{ij} - 1) \geq \big(\lceil \frac{1}{\delta} \rceil \delta - 1 \big) \geq 0$. Thus the second term in the LHS of the above equation is non-negative. It follows that  the $\ell_1$ norm of any ICS is greater than that of any full support solution of the LCP$_\delta(G)$.
\end{proof}

\begin{theorem}\label{thm:cliqmain}
Any solution of {\rm maxICS}$_\delta(G)$ is  also a solution of {\rm maxSOL}$_\delta(G)$ for any $\delta \in [\eta(G),1)$.
\end{theorem}
\begin{proof}
First, we note from \cref{prop:cliqfull} that the $\ell_1$ norm of any ICS corresponding to maximum independent sets is at least as much as that of any solution with full support. Thus, we only need to prove that the $\ell_1$ norm of solution of {\rm maxICS}$_\delta(G)$ is at least as much as that of any other solution which does not have full support.

Consider a solution $x$ such that $\sigma(x) \subset V$ and assume that it maximizes the $\ell_1$ norm. Then, consider the restricted graph $G_{\sigma(x)}$. We have that in this graph, $x$ is a full support solution. Thus, $\norm{x}_1 \leq ||x_{\max}||_1$, where $x_{\max}$ is a solution of {\rm maxICS}$_\delta(G_{\sigma(x)})$. Also, by  \cref{lem:cliqorder}, we get  $||x_{\max}||_1 \leq ||x'_{\max}||_1$, where $x'_{\max}$ is a solution of {\rm maxICS}$_\delta(G)$. Thus, the $\ell_1$ norm of $x$ is less than or equal to the $\ell_1$ norm of any solution of {\rm maxICS}$_\delta(G)$ for any $\delta \in [\eta(G),1)$. 
\end{proof}

\subsection{Graphs with a Unique Maximum Independent Set}\label{sec:k-ind}

For a graph $G=(V,E)$ set $I\subseteq V$ is said to be the {\it  unique maximum independent set} if there is no other indepdendent set $I'\subseteq V$ with $|I'| \geq |I|$. In this section, we show that when a unique independent set ($S$) exists for a graph $G$, then $x = \bfone_S$ is a solution of {\rm maxSOL}$_\delta(G)$ for $\delta \in [\eta(G),1)$. 

\begin{corollary}\label{cor:k-ind2}
For a graph $G$, when a unique maximum independent set $S$ exists, the maximum $\ell_1$ norm amongst the solutions of  \textup{LCP}$_\delta(G)$ is achieved by the characteristic vector of $S$ for $\delta \in [\eta(G),1)$.
\end{corollary}
\begin{proof}
From \cref{thm:cliqmain} and \cref{prop:cliqmaxind}, we know that there exists an ICS $x^\mathcal{K}(\delta)$ which solves maxSOL$\delta(G)$ where $\mathcal{K}$ is a union of disjoint independent cliques and $|\mathcal{K}| = \alpha(G)$. Let $\mathcal{K} = \{C_1,\dots,C_{\alpha(G)}\}$. We claim that all the cliques $\{C_1,\dots,C_{\alpha(G)}\}$ are $K_1$'s (single vertices). Suppose not, i.e. suppose for some $j \in \{1,\dots,\alpha(G)\}$, $C_j$ has at least two vertices (say $j_1$ and $j_2$). Then, we can construct two distinct independent sets ($S_1$ and $S_2$) by choosing one vertex from each of $C_i$ for $i \in \{1,\dots,\alpha(G)\} \backslash j$ and choosing $j_1$ in $S_1$ and $j_2$ in $S_2$ with $|S_1| = |S_2| = \alpha(G)$. But this violates the definition of unique maximum independent sets. Thus, each $C_j$ for $j \in \{1,\dots,\alpha(G)\}$ must have only one vertex and thus $\mathcal{K}$ itself is the unique maximum independent set. Hence, we have that $\bfone_\mathcal{S}$ is a solution of \textup{LCP}$_\delta(G)$ which maximizes the $\ell_1$ norm.
\end{proof}

\section{Solutions of LCP$_\delta(G)$ for $\delta \geq 1$} \label{sec:geq1}

The case when $\delta = 1$ has been extensively studied in \cite{pandit2018linear}, wherein they show that the weighted independence number of the graph is also the maximum weighted $\ell_1$ norm over the solution set of a linear complementarity problem (LCP(G)).

The case when $\delta > 1$ is very similar to the case when $\delta  = 1$. \cref{lem:gen} (e) shows that in this case the support of any solution of LCP$_\delta(G)$ is a $1$-dominating set. Also, \cref{lem:intsol} shows that $x$ is an integer solution of LCP$_\delta(G)$ if and only if it is the characteristic vector of a maximal independent set. We state the following theorem similar to \cref{thm:parthe1} for $\delta \geq 1$. The proof follows a very similar line of reasoning like Theorem 1 of \cite{pandit2018linear}.

\begin{theorem}\label{thm:geq1}
For the \textup{LCP}$_\delta(G), \ \delta > 1$, the maximum weighted $\ell_1$ norm is achieved at the characteristic vector of a  maximum weighted independent set, i.e. 
\[
\alpha_w(G) = \max\{w\t x \  | \  x \ {\rm  solves } \ \textup{LCP}_\delta(G) \},
\]
where $\alpha_w(G)$ represents the weighted independence number with weights $w$.
\end{theorem}
\begin{proof}
Let $G$ be a graph such that $|V(G)| = n$. Let $w = (w_1,w_2,\hdots,w_n)$ denote the vector of given weights. We define $M_w(\delta,G) = \max\{w\t x \  | \  x \ {\rm  solves } \ \LCP_\delta(G) \}$. We prove $M_w(\delta,G) = \alpha_w(G)$ by showing inequalities in both directions.

Let $S$ be a $w$-weighted maximum independent set. Then, it is a maximal independent set and $\bfone_S$ solves $\LCP_\delta(G)$ for $\delta \geq 1$ by \cref{lem:intsol}. Thus, $M_w(\delta,G) \geq w\t\bfone_S = \alpha_w(G)$.

Next, we want to show that $M_w(\delta,G) \leq \alpha_w(G)$. We prove this by induction on the size of the graph. For the base case, consider the graph with $1$ vertex and no edges. In this case, it is trivial to see that $M_w(\delta,G) \leq \alpha_w(G)$.\\

\noindent {\it Induction Hypothesis:} For any graph $G'$ with $|V(G')| < n$, for a given set of weights $w$ and $\delta \geq 1$, we have $M_w(\delta,G') \leq \alpha_w(G')$.

Let $G$ be a graph with vertex set $V(G) = \{1,2,\hdots,n\}$ and let $x^*$ be the maximizer of $w\t x$ where $x \in \SOL_\delta(G)$. Let $S$ be a $w$-weighted maximum indepdendent set of $G$.

\noindent {\it Case 1:} $\sigma(x^*) = V(G)$\\
Here, we have $x^*_i > 0 \ \forall i \in V(G)$. Thus, the LCP conditions \cref{eqn:lcp1,eqn:lcp2,eqn:lcp3} dictate 
\[
\C_i(x^*) = x^*_i + \delta \sum_{j=1}^{n}a_{ij}x^*_j = 1 \ \forall i \in V(G).
\]
Taking the $w$-weighted sum over $i \in S$, we get
\begin{equation}\label{eqn:geqdelta1}
 \sum_{i \in S}w_ix^*_i = \sum_{i \in s}w_i - \delta\sum_{i \in S}w_i\sum_{j=1}^{n}a_{ij}x^*_j.
\end{equation}
Now,
\begin{align} \nonumber
 M_w(\delta,G) 
 &= \sum_{i \in V(G)}w_ix^*_i \\ \nonumber
 &= \sum_{i \in S}w_ix^*_i + \sum_{i \in V(G) \backslash S}w_ix^*_i \\ \label{eqn:geqdelta2}
 &= \sum_{i \in S}w_i - \delta\sum_{i \in S}w_i\sum_{j=1}^{n}a_{ij}x^*_j + \sum_{i \in V(G) \backslash S}w_ix^*_i \\ \label{eqn:geqdelta3}
 &= \alpha_w(G) - \sum_{j \in V(G) \backslash S}x^*_j(\delta\sum_{i \in S}w_i a_{ij} - w_j),
\end{align}
where \cref{eqn:geqdelta2} follows by using \cref{eqn:geqdelta1} and \cref{eqn:geqdelta3} follows because $a_{ij} = 0$ when $i,j \in S$. We claim that $\sum_{i \in S}w_i a_{ij} - w_j \geq 0 \ \forall j \in V(G) \backslash S$. Suppose not, then we construct an independent set with higher weight than that of $S$. Let $l \in V(G) \backslash S$ be such that $w_l - \sum_{i \in S}w_i a_{il} > 0$. Consider $\widetilde{S} = (S \backslash N_G(l)) \cup {l}$. It is easy to see that $\widetilde{S}$ is an independent set and the difference between weights of $\widetilde{S}$ and $S$ is $w_l - \sum_{i \in S}w_i a_{il} > 0$. This contradicts the fact that $S$ is a $w$-weighted maximum independent set. Thus, $\sum_{i \in S}w_i a_{ij} - w_j \geq 0 \ \forall j \in V(G) \backslash S$. Since $\delta \geq 1$, this implies $\delta \sum_{i \in S}w_i a_{ij} - w_j \geq 0 \ \forall j \in V(G) \backslash S$ which gives $M_w(\delta,G) \leq \alpha_w(G)$.\\

\noindent {\it Case 2}: $\sigma(x^*) \subset V(G)$, a strict subset\\
From \cref{lem:gen}, we know that since $x^* \in \SOL_\delta(G)$, we have $x^*_{\sigma(x^*)} \in \SOL_\delta(G_{\sigma(x^*)})$. 
For brevity let $y^*:=x^*_{\sigma(x^*)}$. Now, 
\[
 M_w(\delta,G) =  \sum_{i \in V(G)}w_i x^*_i = \sum_{i \in V(G_{\sigma(x^*)})}w_i y^*_i \leq M_w(\delta,G_{\sigma(x^*)}) \leq \alpha_w(G_{\sigma(x^*)}) \leq \alpha_w(G),
\]
where $M_w(\delta,G_{\sigma(x^*)}) \leq \alpha_w(G_{\sigma(x^*)})$ follows from the induction hypothesis since \\ $|V(G_{\sigma(x^*)})| < n$. Hence, by the principle of mathematical induction, we have that for any graph $G$, with given weights $w$ and $\delta \geq 1$, we have $M_w(\delta,G) \leq \alpha_w(G)$. 
\end{proof}

Thus, we have proved that for the case $\delta \geq 1$, the weighted $\ell_1$ norm amongst the set SOL$_\delta(G)$ is maximized at the characteristic vector of the maximum weighted independent set. Particularly, the $\ell_1$ norm is maximized by maximum independent sets and they are also independent clique solutions.

\section{Conclusion}\label{sec:con}
In this paper, we identified the $\ell_1$ norm maximizing solutions of the LCP$_\delta(G)$. We introduced a new notion of independent clique solutions. These are generalizations of independent sets with independent sets being the special case in which each clique is degenerate. We showed that for $\delta \geq \eta(G)$, the $\ell_1$ norm maximizing ICSs are also the $\ell_1$ norm maximizing solutions of LCP$_\delta(G)$. 

The authors of \cite{bramoulle2014strategic} have shown that for $\delta < \frac{1}{|\lambda_{\min}(A)|}$, where $\lambda_{\min}(G)$ is the lowest eigenvalue of the adjacency matrix of graph $G$, LCP$_\delta(G)$ has a unique solution. Thus, that solution itself is $\ell_1$ norm maximizing. Our results identify certain $\ell_1$ norm maximizing solutions for $\delta \geq \eta(G)$. For $\delta \in (\frac{1}{|\lambda_{\min}(A)|}, \eta(G))$, the question of which solutions of LCP$_\delta(G)$ maximize the $\ell_1$ norm remains open.

\appendix

\section{Proofs from \cref{sec:lcpg}}

\subsection{Proof of \cref{lem:gen}}\label{appen:lem:gen}
~\\
\begin{enumerate}[label = (\alph*)]
\item From \cref{eqn:nbdsum}, $\C_i(0) = 0$ which implies \cref{eqn:lcp2} is violated. Thus, $0 \not \in$ SOL$_\delta(G$).

\item From \cref{eqn:nbdsum}, $\C_i(x) =  x_i + \delta \sum_{j \in N_G(i)}x_j$. Since $\delta \geq 0$ and $x_i \geq  0 \ \forall i \in V(G)$ from \cref{eqn:lcp1}, $\C_i(x) \geq x_i$. Thus,  $\C(x) \geq x \ \forall x \in \textup{SOL}_\delta(G$).

\item We know that $x \geq 0$. Let $x_i > 0$ for some $i$. Then, from \cref{eqn:lcp3}, $\C_i(x)  = 1$. From part (b) of \cref{lem:gen}, $x_i \leq \C_i(x) = 1$. Thus, $x_i \in [0,1] \ \forall i \in V(G)$.

\item First, consider $x^1 \in \textup{SOL}_\delta(G_1)$ and $x^2 \in \textup{SOL}_\delta(G_2)$, we prove that $x = (x^1,x^2) \in \textup{SOL}_\delta(G)$. Since $x^1,x^2  \geq 0$ we have $x \geq 0$. Next, note that since $G_1$ and $G_2$ are disjoint, $N_{G_1}(i) = N_G(i) \  \forall i \in V(G_1)$ and $N_{G_2}(i) = N_G(i) \ \forall i \in V(G_2)$. Thus, for $i \in V(G_1)$, $\C_i(x) = x_i +  \delta \sum_{j \in N_G(i)}x_j = x^1_i +  \delta \sum_{j \in N_{G_1}(i)}x^1_j = \C_i(x^1)$. Similarly,  for $i \in V(G_2)$, $\C_i(x) =  \C_i(x^2)$. Thus we get $\C_i(x) \geq 1 \ \forall i \in V(G)$. Also, $x^1_i(\C_i(x^1) - 1) = 0 \ \forall i \in V(G_1)$ and $x^2_i(\C_i(x^2) - 1) = 0 \ \forall i \in V(G_2)$  imply $x_i(\C_i(x) - 1) = 0 \ \forall i \in V(G)$. Thus, $x \in  \textup{SOL}_\delta(G)$.

Next, given $x \in \textup{SOL}_\delta(G)$, we show that $x_{G_1} \in \textup{SOL}_\delta(G_1)$ and $x_{G_2} \in \textup{SOL}_\delta(G_2)$, where $x_{G_1} \in\Real^{|V(G_1)|}$, $x_{G_2} \in \Real^{|V(G_2)|}$ and 
$x=(x_{G_1},x_{G_2})$. 
Clearly $x_{G_1} \geq 0$.  Using $N_{G_1}(i) = N_G(i) \  \forall i \in V(G_1)$, we get $\C_i(x) =  \C_i(x_{G_1})$ and thus $x_{G_1}$  satisfies both \cref{eqn:lcp2,eqn:lcp3}. Hence, $x_{G_1} \in \textup{SOL}_\delta(G_1)$. Similarly, $x_{G_2} \in \textup{SOL}_\delta(G_2)$.

\item From \cref{lem:gen} (c), $0 \leq x_i \leq 1 \ \forall i \in V(G)$. Thus $\bfone_{\sigma(x)} \geq x$. Thus, for all $i \in V(G)$, $\C_i(\bfone_{\sigma(x)}) \geq \C_i(x) \geq 1$. Let $i \not \in \sigma(x)$. For this $i$ we have, $\C_i(\bfone_{\sigma(x)}) = 0 + \delta\sum_{j \in N_G(i)}x_j = \delta |N_G(i) \cap \sigma(x)|$. Thus, $|N_G(i) \cap \sigma(x)| \geq \frac{1}{\delta}$. Since, $|N_G(i) \cap \sigma(x)|$ is a positive integer, $|N_G(i) \cap \sigma(x)| \geq \lceil \frac{1}{\delta} \rceil$. Thus,  $\sigma(x)$ is a $\lceil \frac{1}{\delta} \rceil $-dominating set of $G$.

\item For $i \in V(G_{\sigma(x)})$, $\tilde{x}_i = x_i > 0$. Thus, $\sigma{(\tilde{x})} = V(G_{\sigma(x)})$. Note that for a vertex $i$ in $G_{\sigma(x)}$, the discounted sum of the closed neighborhood $\tilde{\C}(\tilde{x}) = \tilde{x}_i + \delta\sum_{j \in N_{G_{\sigma(x)}}(i)}\tilde{x}_j = x_i + \delta\sum_{j \in N_G(i)}x_j = \C_i(x)$ since $x_i = 0$ for $i \not \in \sigma(x)$. Hence,  $\tilde{x} = x_{\sigma(x)} \in \textup{SOL}_\delta(G_{\sigma(x)})$.
\end{enumerate}

\section{Proofs from \cref{sec:les1}}

\subsection{Proof of Two Clique Lemma (\cref{lem:cont0})}\label{appen:lem:cont0}
We will show that for the $x$ given the statement of \cref{lem:cont0}, $C_i(x) \geq 1$ for the interval $\delta \in [\gamma(G),1)$. It will the follow that for the given $x$ with $x_i = 0$, the LCP conditions \cref{eqn:lcp1,eqn:lcp2,eqn:lcp3} are satisfied. Since $i$ is connected to all vertices of $C_1$ and at least one vertex of $C_2$,
\begin{align*}
C_i(x) &\geq \delta \Bigg(\frac{n}{1+(n-1)\delta} + \frac{1}{1+(m-1)\delta}\Bigg)
\end{align*}
Now, observe that
\begin{align*}
& &\delta\Bigg(\frac{n}{1+(n-1)\delta} + \frac{1}{1+(m-1)\delta}\Bigg) &\geq 1 
%&\iff & \delta n + n(m-1)\delta^2 + \delta+(n-1)\delta^2 & \geq (1+(n-1)\delta)(1+(m-1)\delta) \\
\iff & (m + n - 2)\delta^2 + (3-m)\delta -1  \geq 0.
\end{align*}
 This in turn is true in some range $\delta \in \gamma(G),1)$ if and only if the positive root, denoted $\gamma(m,n)$, of the above quadratic in $\delta$ is less than 1. We know,
\[
\gamma(m,n) = \frac{m-3 + \sqrt{(3-m)^2 + 4(m + n - 2)}}{2(m+n-2)}
\]
Now, it is easy to check that 
\begin{align*}
& & \gamma(m,n)& \leq 1 &\iff &&   m+n \geq 2,
\end{align*}
which is true. 
We have $\gamma(m,n)$ as a decreasing function in $n$ and increasing in $m$. Thus, the largest value of $\gamma(m,n)$ is attained at $n = 1$ and $m = \omega(G)$. Thus we get the result is true for $\delta \in [\gamma(G),1)$ where \[\gamma(G) :=\gamma(\omega(G),1)= \frac{\omega(G)-3 + \sqrt{(\omega(G)-3)^2 + 4(\omega(G)-1)}}{2(\omega(G)-1)}.\]
This is as required, thereby completing the proof.

\bibliographystyle{siamplain}
\bibliography{ref}

\end{document}